\documentclass{article}

\usepackage[margin=0.5in]{geometry}
\usepackage[utf8]{inputenc}
\usepackage{cancel}
\usepackage{csquotes}
\usepackage{lipsum}
\usepackage{amsfonts}
\usepackage{amsthm,amsmath,amssymb,accents,empheq} 
\usepackage{enumitem}
\usepackage{subfig}
\usepackage{graphicx}
\usepackage{float}
\usepackage{xcolor}
\usepackage{algorithm}
\usepackage[noend]{algpseudocode}
\usepackage{hyperref}

\hypersetup{
    colorlinks=true,
    linkcolor=blue,
    citecolor=blue,
    filecolor=blue,      
    urlcolor=blue,
}

\usepackage[
backend=biber,
style=ieee,
citestyle=numeric,
sorting=none
]{biblatex}

\newtheorem{theorem}{Theorem}[section]
\newtheorem{lemma}{Lemma}[section]

\theoremstyle{remark}

\newcommand{\myparallel}{{\mkern3mu\vphantom{\perp}\vrule depth 0pt\mkern2mu\vrule depth 0pt\mkern3mu}}

\usepackage{listings}
 
\definecolor{codegreen}{rgb}{0,0.6,0}
\definecolor{codegray}{rgb}{0.5,0.5,0.5}
\definecolor{codepurple}{rgb}{0.58,0,0.82}
\definecolor{backcolour}{rgb}{0.95,0.95,0.92}
 
\lstdefinestyle{mystyle}{
    backgroundcolor=\color{backcolour},   
    commentstyle=\color{codegreen},
    keywordstyle=\color{magenta},
    numberstyle=\tiny\color{codegray},
    stringstyle=\color{codepurple},
    basicstyle=\ttfamily\footnotesize,
    breakatwhitespace=true,         
    breaklines=true,                 
    captionpos=b,                    
    keepspaces=true,                 
    numbers=left,                    
    numbersep=5pt,                  
    showspaces=false,                
    showstringspaces=false,
    showtabs=false,                  
    tabsize=2
}

\usepackage{titlesec}
\titleformat{\paragraph}
{\normalfont\normalsize\bfseries}{\theparagraph}{1em}{}
\titlespacing*{\paragraph}
{0pt}{3.25ex plus 1ex minus .2ex}{1.5ex plus .2ex}

\title{\textbf{A Particle-in-cell Method for Plasmas with a Generalized Momentum Formulation, \\ Part II: Enforcing the Lorenz Gauge Condition} \thanks{The research of the authors was supported by AFOSR grants FA9550-19-1-0281 and  FA9550-17-1-0394, NSF grant DMS-1912183, and DOE grant DE-SC0023164.}} 

\author{Andrew J. Christlieb \thanks{Department of Computational Mathematics, Science and Engineering, Michigan State University, East Lansing, MI, 48824, United States; \href{mailto:christli@msu.edu}{christli@msu.edu}.}
\and William A. Sands \thanks{Department of Mathematical Sciences, University of Delaware, Newark, DE, 19716, United States;
\href{mailto:wsands@udel.edu}{wsands@udel.edu} (corresponding author).}
\and Stephen White \thanks{Department of Computational Mathematics, Science and Engineering, Michigan State University, East Lansing, MI, 48824, United States; \href{mailto:whites73@msu.edu}{whites73@msu.edu}.}
}

\date{}

\begin{document}

\maketitle

\begin{abstract}
     In a previous paper \cite{christlieb2024pic}, we developed a new particle-in-cell (PIC) method for the relativistic Vlasov-Maxwell system in which the electromagnetic fields and the equations of motion for the particles were cast in terms of scalar and vector potentials through a Hamiltonian formulation. This new method evolved the potentials under the Lorenz gauge using integral equation methods. New methods to construct spatial derivatives of the potentials that converge at the same rates as the fields were also presented. The new particle method was compared against standard explicit discretizations, including the well-known FDTD-PIC method, for a range of applications involving sheaths and particle beams. This paper extends this new class of methods by focusing on the enforcement the Lorenz gauge condition in both exact and approximate forms using co-located meshes. A time-consistency property of the proposed field solver for the vector potential form of Maxwell's equations is established, which is shown to preserve the equivalence between the semi-discrete Lorenz gauge condition and the analogous semi-discrete continuity equation. Using this property, we present three methods to enforce a semi-discrete gauge condition. The first method introduces an update for the continuity equation that is consistent with the discretization of the Lorenz gauge condition. Both the finite difference and spectral implementations satisfy this discrete gauge condition to machine precision.  The second approach we propose enforces a semi-discrete continuity equation using the boundary integral solution to the field equations. The potential benefit of this approach is that it eliminates spatial derivatives that appear on the particle data, namely the current density, which is often calculated by linear combinations of low-order spline basis functions. This method is ideally suited to boundary integral equation methods that invert multi-dimensional operators without dimensional splitting techniques and will be the subject of future work. The third approach introduces a gauge correcting method that makes direct use of the gauge condition to modify the scalar potential and uses local maps for both the charge and current densities. This results in a gauge error, as the maps do not enforce the continuity equation.  The vector potential coming from the current density is taken to be exact, and using the Lorenz gauge, we compute a correction to the scalar potential that makes the two potentials satisfy the gauge condition. This method also enforces the gauge condition to machine precision. We demonstrate two of the proposed methods in the context of periodic domains. Problems defined on bounded domains, including those with complex geometric features remain an ongoing effort. However, this work shows that it is possible to design computationally efficient methods that can effectively enforce the Lorenz gauge condition in an non-staggered PIC formulation.
\end{abstract}


%
%
%
%
\section{Introduction}

Recently, in \cite{christlieb2024pic}, we developed a new particle-in-cell (PIC) method for the Vlasov-Maxwell (VM) system in which Maxwell's equations were cast as wave equations under the Lorenz gauge and solved using an integral equation method. New methods for the spatial derivatives were also obtained by computing analytical derivatives of the resulting integral solutions. The proposed field solver, as well as the methods used to compute the derivatives, converge at a rate of fifth-order accuracy in space and first-order accuracy in time. These methods were coupled to a particle discretization of the Vlasov equation in its Hamiltonian form, which was evolved using a Taylor-corrected semi-implicit Euler method. The new approach demonstrated a number of improvements over conventional explicit PIC methods, including mesh-independent heating and improved stability and refinement, even in cases where the Debye length $\lambda_{D}$ would normally be considered under-resolved (i.e., $\lambda_{D}/\Delta x \geq 1$). These characteristics seem to imply that the new method requires less resolution than traditional explicit PIC in terms of macroparticles and mesh spacing for a prescribed accuracy. Though not explored in this paper, the new method also retains the geometric flexibility provided by the field solver \cite{causley2017wave-propagation,MOLT-EB-2020}, and will be important in forthcoming research as we consider more realistic devices with non-trivial geometries. In electromagnetic experiments, the error in the Lorenz gauge condition was carefully monitored, given that no method was used to enforce the gauge condition. We found these errors to be acceptable even in problems known to be sensitive to issues of charge conservation. 
However, we anticipate situations where non-linear dynamics of the fields interacting with a non-thermalized plasma will lead to settings where this is not the case, and one will need to enforce the gauge condition.

This paper extends the method introduced in part I by developing different approaches to enforce the Lorenz gauge condition in the context of co-located meshes. We start by establishing a certain time-consistency theorem, which shows, at the semi-discrete level, that enforcing the Lorenz gauge is equivalent to satisfying the continuity equation. Using the time-consistency result from this paper, we first develop a map for the charge density that ensures we exactly satisfy a semi-discrete analogue of the Lorenz gauge condition on a non-staggered grid. This method utilizes bilinear map of current density on the mesh, which is then coupled to a point-wise solution to the continuity equation to compute the charge density. In this time-consistent approach, the spatial derivatives are computed with high-order finite differences (e.g., sixth-order accuracy) or with a spectral method. A Lagrange multiplier is introduced to enforce charge conservation globally. When this charge density is used in the update of the scalar potential, both with and without the Lagrange multiplier, the method effectively satisfies the Lorenz gauge condition to machine precision. The second approach we propose is an exact map that enforces a semi-discrete continuity equation through a source term in the integral formulation of the field solver. A notable advantage offered by the integral solution with this particular map is that it allows for the removal of spatial derivatives that act on the particle current density. This quantity may be ``rough" if too few simulation particles are used, as it is computed from linear combinations of low-order spline basis functions. The action of the spatial derivatives further amplifies this. Instead, we propose a new map that eliminates the spatial derivatives of the particle data by exploiting the structure of the integral solution for the scalar potential. The method is ideal for integral solutions that make use of the multi-dimensional Green's function as in \cite{cheng2017asymptotic}.  The third method we consider is a gauge correcting method in which bilinear mappings are used for both charge and current density on the mesh. Since the resulting fields will not satisfy the gauge condition, we treat the vector potential solution as exact and use the Lorenz gauge to compute a correction to the scalar potential. The vector potential and corrected scalar potential are then guaranteed to satisfy the gauge condition to machine precision. 
 
 One of the major goals of this effort is to design fast algorithms that can address the challenges posed by the geometry of experimental devices. For this reason, we focus on simple mappings that are based on bilinear maps (or area weightings) which are more convenient in such circumstances and offer a good balance between accuracy and efficiency. Bilinear maps produce smoother representations of the particle data on the mesh compared to other maps such as piece-wise constant maps or the current weighting scheme of Villasenor and Buneman \cite{VillasenorChargeConservation92}. For the gauge enforcing methods, the exact charge density mapping proposed in this work uses bilinear maps only for the current densities, and the continuity equation is used to identify the corresponding charge density. The gauge correcting approach, on the other hand, uses bilinear maps for both the charge and current densities, but the gauge error is controlled through an equation that corrects the scalar potential.

This structure of this paper is as follows. In section \ref{sec:2 formulation}, we give a short introduction to the formulation considered in this paper, which was used in our previous paper \cite{christlieb2024pic}. We also establish some important properties of the semi-discrete formulation used for the potentials, including a time-consistency property of the field solver, which largely motivates the methods introduced in subsequent sections. We then introduce three techniques that enforce a semi-discrete gauge condition in section \ref{sec:3 exact charge maps}. Some numerical results are presented in section \ref{sec:5 Numerical results} to test the methods in periodic domains. In particular, we consider the relativistic Weibel instability and present a new test problem that is designed to amplify the errors in the gauge condition. A brief summary of the paper is presented in section \ref{sec:6 Conclusion}.

%
%
%
%
\section{Formulation and Background}
\label{sec:2 formulation}

This section contains a description of the problem formulation for the Vlasov-Maxwell system that is considered in this work, which is introduced in \ref{subsec:2 VM-system}. A Hamiltonian formulation for this problem is presented in section \ref{subsec:2 model}, which uses the Lorenz gauge to cast the system in terms of scalar and vector potentials. Under the Lorenz gauge, these equations take the form of scalar wave equations, which are evolved using a backwards difference discretization in time. For purposes of brevity, we provide only the essential details concerning the formulation. The full details concerning the formulation and the dimensionally-split algorithms can be found in our previous paper \cite{christlieb2024pic}. At the end of the section, we conclude with a brief summary of the important elements of the formulation.

%
%
\subsection{Relativistic Vlasov-Maxwell System}
\label{subsec:2 VM-system}

In this work, we develop numerical algorithms for plasmas described by the relativistic Vlasov-Maxwell (VM) system, which in SI units, reads as
\begin{empheq}[left=\empheqlbrace]{align}
    &\partial_t f_{s} + \frac{\mathbf{p}}{m_{s}\gamma_{s}} \cdot \nabla_{x} f_{s} + q_{s} \left( \mathbf{E} + \frac{\mathbf{p} \times \mathbf{B}}{m_{s}\gamma_{s}} \right)\cdot \nabla_{p} f_{s} = 0, \label{eq:Vlasov species}\\
    &\nabla \times \mathbf{E} = -\partial_t \mathbf{B}, \label{eq:Farady} \\
    &\nabla \times \mathbf{B} = \mu_{0}\left( \mathbf{J} + \epsilon_{0} \partial_t \mathbf{E} \right), \label{eq:Ampere} \\
    &\nabla \cdot \mathbf{E} = \frac{\rho}{\epsilon_0}, \label{eq:Gauss-E} \\
    &\nabla \cdot \mathbf{B} = 0. \label{eq:Gauss-B}
\end{empheq}
The first equation \eqref{eq:Vlasov species} is the relativistic Vlasov equation which describes the evolution of a probability distribution function $f_{s}\left( \mathbf{x}, \mathbf{p}, t\right)$ for particles of species $s$ in phase space which have mass $m_{s}$ and charge $q_{s}$. Here, we define $\gamma_{s} = 1/\sqrt{1 + \mathbf{p}^2/(m_{s} c)^2} =  1/\sqrt{1 - \mathbf{v}^2/c^2}$, which makes equation \eqref{eq:Vlasov species} Lorentz invariant. Physically, equation \eqref{eq:Vlasov species} describes the time evolution of a distribution function that represents the probability of finding a particle of species $s$ at the position $\mathbf{x}$, with linear momentum $\mathbf{p} = m_{s} \gamma_{s} \mathbf{v}$, at any given time $t$. Since the position and velocity data are vectors with 3 components, the distribution function is a scalar function of 6 dimensions plus time. While the equation itself has fairly simple structure, the primary challenge in numerically solving this equation is its high dimensionality. This growth in the dimensionality has posed tremendous difficulties for grid-based discretization methods, where one often needs to use many grid points to resolve relevant space and time scales in the problem. This difficulty is compounded by the fact that many plasmas of interest contain multiple species. Despite the lack of a collision operator on the right-hand side of \eqref{eq:Vlasov species}, collisions occur in a mean-field sense through the electric and magnetic fields, which appear as coefficients of the gradient in momentum.

Equations \eqref{eq:Farady} - \eqref{eq:Gauss-B} are Maxwell's equations, which describe the evolution of the background electric and magnetic fields. Since the plasma is a collection of moving charges, any changes in the distribution function for each species will be reflected in the charge density $\rho(\mathbf{x},t)$, as well as the current density $\mathbf{J}(\mathbf{x},t)$, which, respectively, are the source terms for Gauss' law \eqref{eq:Gauss-E} and Amp\`ere's law \eqref{eq:Ampere}. For $N_{s}$ species, the total charge density and current density are defined by summing over the species
\begin{equation}
    \rho(\mathbf{x},t) = \sum_{s=1}^{N_{s}} \rho_{s}(\mathbf{x},t), \quad \mathbf{J}(\mathbf{x},t) = \sum_{s=1}^{N_{s}} \mathbf{J}_{s}(\mathbf{x},t), \label{eq:total charge + current densities}
\end{equation}
where the species charge and current densities are defined through moments of the distribution function $f_s$:
\begin{equation}
    \rho_{s}(\mathbf{x},t) = q_{s} \int_{\Omega_{p}} f_{s}(\mathbf{x}, \mathbf{p}, t) \, d\mathbf{p}, \quad \mathbf{J}_{s}(\mathbf{x},t) = q_{s} \int_{\Omega_{p}} \frac{\mathbf{p}}{m_{s} \gamma_{s}} f_{s}(\mathbf{x}, \mathbf{p}, t) \, d\mathbf{p}. \label{eq:species charge + current densities integrals}
\end{equation}
Here, the integrals are taken over the momentum components of phase space, which we have denoted by $\Omega_{p}$. The values $\rho$ and $\textbf{J}$ are physically coupled by the continuity equation
\begin{equation}\label{eq:continuity}
    \partial_t \rho + \nabla \cdot \textbf{J} = 0,
\end{equation}
which enforces the conservation of charge.

The remaining parameters $\epsilon_{0}$ and $\mu_{0}$ describe the permittivity and permeability of the media in which the fields propagate, which we take to be free-space. In free-space, Maxwell's equations move at the speed of light $c$, and we have the useful relation $c^{2} = 1/\left( \mu_0 \epsilon_0 \right)$. The last two equations, \eqref{eq:Gauss-E} and \eqref{eq:Gauss-B}, are known as the involutions.  The involutions are passive equations, which are not directly part of the update.  If initial conditions satisfy the involutions, then solutions to \eqref{eq:Farady} and \eqref{eq:Ampere} satisfy  \eqref{eq:Gauss-E} and \eqref{eq:Gauss-B} for all future times.  It is imperative that numerical schemes for Maxwell's equations satisfy these involutions to ensure charge conservation in the sense of equation \eqref{eq:continuity} and prevents the appearance of so-called ``magnetic monopoles". This requirement is one of the reasons we adopt a gauge formulation for Maxwell's equations, which is presented in the next section.

%
%
%
%
\subsection{A Hamiltonian Formulation of the Vlasov-Maxwell System under the Lorenz Gauge}
\label{subsec:2 model}

The particle method considered in the companion paper \cite{christlieb2024pic} is a Hamiltonian formulation of the Vlasov-Maxwell system in which Maxwell's equations were written in terms of the scalar and vector potentials. The adoption of a Hamiltonian formulation for the particles was motivated by the desire to avoid computing time derivatives of the vector potential. Using the Hamiltonian for a single relativistic particle in a potential field
\begin{equation}
    \label{eq:Relativistic Hamiltonian particle potential field}
    \mathcal{H} = \sqrt{c^2 \left( \mathbf{P} - q \mathbf{A}\right)^2 + \left(mc^2\right)^2} + q \phi,
\end{equation}
we obtained the system
\begin{empheq}[left=\empheqlbrace]{align}
    &\frac{d\mathbf{x}_{i}}{dt} = \frac{c^2 \left(\mathbf{P}_{i} - q_{i}\mathbf{A}\right)}{\sqrt{c^2\left(\mathbf{P}_{i} - q_{i} \mathbf{A}\right)^2 + \left(m_{i} c^2\right)^2}}, \label{eq:Position equation relativistic form} \\
    &\frac{d\mathbf{P}_{i}}{dt} = -q_{i} \nabla \phi + \frac{q_{i} c^2 \left( \nabla\mathbf{A}\right) \cdot \left(\mathbf{P}_{i} - q_{i}\mathbf{A}\right)}{\sqrt{c^2\left(\mathbf{P}_{i} - q_{i} \mathbf{A}\right)^2 + \left(m_{i} c^2\right)^2}}, \label{eq:Generalized momentum equation relativistic form} \\
    &\frac{1}{c^2} \partial_{tt} \phi -\Delta \phi =  \frac{1}{\epsilon_{0}} \rho, \label{eq:scalar potential eqn lorenz} \\ 
    &\frac{1}{c^2} \partial_{tt} \mathbf{A} -\Delta \mathbf{A}= \mu_{0} \mathbf{J}, \label{eq:vector potential eqn lorenz} \\ 
    &\frac{1}{c^2} \partial_{t} \phi + \nabla \cdot \mathbf{A} =0 \label{eq:Lorenz gauge condition},
\end{empheq}
where $i = 1, 2, \cdots, N_{p}$ is the index of the simulation particle, $c$ is the speed of light, $\epsilon_{0}$ and $\mu_{0}$ represent, respectively, the permittivity and permeability of free-space. Further, we have used $\phi$ to denote the scalar potential and $\mathbf{A}$ is the vector potential, which are related to the electric and magnetic fields via
\begin{equation}
    \label{eq:Convert potentials to EB}
    \mathbf{E} = -\nabla \phi - \partial_{t} \mathbf{A}, \quad \mathbf{B} = \nabla \times \mathbf{A}.
\end{equation}
Though not written explicitly, the potentials and gradients that appear in equations \eqref{eq:Position equation relativistic form} and \eqref{eq:Generalized momentum equation relativistic form} for particle $i$ are evaluated at the corresponding position given by $\mathbf{x}_{i}$. The evolution equations for the potentials \eqref{eq:scalar potential eqn lorenz} and \eqref{eq:vector potential eqn lorenz} are a system of four scalar wave equations that are coupled to the Lorenz gauge condition \eqref{eq:Lorenz gauge condition}. The wave equations for the potentials $\phi$ and $\mathbf{A}$ are equivalent to Maxwell's equations provided the gauge condition is satisfied. While other gauge conditions can be used, we choose to work under the Lorenz gauge simply because it allows for a relatively straightforward application of our previous work developed for scalar wave equations \cite{causley2017wave-propagation,causley2014higher}. However, it may be advantageous to consider other gauge conditions, such as the Coulomb gauge.

In the Hamiltonian formulation presented above, the velocity of particle $i$ is obtained from the generalized momentum using the definition
\begin{equation}
    \mathbf{v}_{i} = \frac{c^2 \left( \mathbf{P}_{i} - q_{i} \mathbf{A}\right)}{\sqrt{ c^2\left( \mathbf{P}_{i} - q_{i} \mathbf{A}\right)^2 + \left(m_{i}c_{i}^{2}\right)^2}}, \quad i = 1, 2, \cdots, N_{p}, \label{eq:velocity from generalized momentum}
\end{equation}
where the potential $\mathbf{A}$ is to be evaluated at the position of particle $i$. Finally, the charge density $\rho$ and current density $\mathbf{J}$ in equations \eqref{eq:scalar potential eqn lorenz} and \eqref{eq:vector potential eqn lorenz} are computed according to
\begin{align}
    \rho \left( \mathbf{x}, t \right) &= \sum_{p=1}^{N_{p}} q_{p} S \left(\mathbf{x} - \mathbf{x}_{p}(t) \right), \label{eq:PIC charge density} \\
    \mathbf{J} \left( \mathbf{x}, t \right) &= \sum_{p=1}^{N_{p}} q_{p} \mathbf{v}_{p}(t) S \left(\mathbf{x} - \mathbf{x}_{p}(t) \right). \label{eq:PIC current density}
\end{align}
Here, $S$ is the particle shape function, which, in addition to gathering the particle data to the mesh, is responsible for interpolating fields on the mesh to the particles. The use of a shape function can be considered as a spatial regularization of the Klimontovich function
\begin{equation*}
    f(\mathbf{x}, \mathbf{v},t) \approx \sum_{p = 1}^{N_{p}} S \left(\mathbf{x} - \mathbf{x}_{p}(t)\right)\delta \left(\mathbf{v} - \mathbf{v}_{p}(t)\right).
\end{equation*}
The definitions \eqref{eq:PIC charge density} and \eqref{eq:PIC current density} follow directly from the evaluation of the moments \eqref{eq:species charge + current densities integrals}, where the reference to the species is replaced with the more generic label $p$ that can be attributed to a particle of any species. In this work, the scatter and gather steps for the current density \eqref{eq:PIC current density} use bilinear maps obtained from tensor products of piece-wise linear spline basis functions. The charge density \eqref{eq:PIC charge density}, on the other hand, will be evaluated using one of two potential options: (1) a map that enforces a semi-discrete continuity equation or (2) bilinear maps obtained from tensor products of the piece-wise linear spline basis function. The latter option is known to violate the continuity equation, but is the basis of the gauge correcting technique discussed in section \ref{subsec:3-corrected_phi}. The choice of bilinear maps provides a good compromise between accuracy and efficiency and allows the method to work on bounded domains, in which the plasma interacts with the boundary.

%
%
%
%
\subsection{A Semi-discrete Method for the Lorenz Gauge Formulation}

We consider the same semi-discrete treatment of the Lorenz gauge formulation as in our previous work \cite{christlieb2024pic}, which we briefly discuss. We adopted the following first-order semi-discrete Backward Difference Formula (BDF) for the scalar wave equation
\begin{equation}
    \label{eq:BDF-1 semi-discrete equation}
    \left( \mathcal{I} - \frac{1}{\alpha^2} \Delta \right) u^{n+1} = \left( 2 u^{n} - u^{n-1} \right) + \frac{1}{\alpha^2} S^{n+1}(\mathbf{x}) + \mathcal{O}\left(\frac{1}{\alpha^3}\right), \quad \alpha := \frac{1}{c \Delta t}.
\end{equation}
Additionally, we approximated $\partial_{t}$ in the potential formulation of Maxwell's equations using a first-order BDF approach, which led to the semi-discrete system
\begin{empheq}[left=\empheqlbrace]{align}
    &\mathcal{L} \left[\phi^{n+1}\right] = 2 \phi^{n} - \phi^{n-1} + \frac{1}{\alpha^2 \epsilon_0} \rho^{n+1}, \label{eq:BDF-1 semi-discrete phi} \\
    &\mathcal{L} \left[\mathbf{A}^{n+1} \right] = 2 \mathbf{A}^{n} - \mathbf{A}^{n-1} + \frac{\mu_0}{\alpha^2} \mathbf{J}^{n+1}, \label{eq:BDF-1 semi-discrete A} \\
    &\frac{\phi^{n+1} - \phi^{n}}{c^{2} \Delta t} + \nabla \cdot \mathbf{A}^{n+1} = 0, \label{eq:semi-discrete Lorenz}
\end{empheq}
where $\mathcal{L}$ is the modified Helmholtz operator that satisfies
\begin{equation}
    \label{eq:modified helmholtz eqn}
    \mathcal{L} \left[ u \right] := \left( \mathcal{I} - \frac{1}{\alpha^2} \Delta \right) u(\mathbf{x}) = S(\mathbf{x}), \quad \mathbf{x} \in \Omega.
\end{equation}
The operator $\mathcal{L}^{-1}$ is formally defined as a boundary integral equation involving a Green's function, namely
\begin{equation}
    \label{eq:u layer potential solution}
    u(\mathbf{x}) = \int_{\Omega} G(\mathbf{x},\mathbf{y}) S(\mathbf{y}) \, dV_{\mathbf{y}} + \int_{\partial \Omega} \left(  \sigma(\mathbf{y}) G(\mathbf{x},\mathbf{y}) + \gamma(\mathbf{y})\frac{\partial G}{\partial \mathbf{n}} \right) \, dS_{\mathbf{y}}.
\end{equation}
where $\sigma(\mathbf{y})$ is the single-layer potential and $\gamma(\mathbf{y})$ is the double-layer potential that are used to enforce boundary conditions on $u(\mathbf{x})$ \cite{FollandBook1995}. Here, we use $G(\mathbf{x},\mathbf{y})$ to denote the Green's function for the operator $\mathcal{L}$. This has several useful properties that will be established in subsection \ref{subsec:BDF-Properties}.

\subsection{The Improved Asymmetrical Euler Method and System Advance}
\label{subsec:IAEM}

The equations of motion for the particles are discretized in time using the improved asymmetric Euler method (IAEM), which was introduced in our previous paper \cite{christlieb2024pic}. It is based on a simple modification of a semi-implicit Euler method used in \cite{Gibbon2017Hamiltonian} that includes an additional term from a Taylor expansion in the momentum advance. A single time step of the IAEM method consists of the following steps:
\begin{empheq}[left=\empheqlbrace]{align}
    \mathbf{x}_{i}^{n+1} &= \mathbf{x}_{i}^{n} + \mathbf{v}_{i}^{n} \Delta t, \label{eq:IAEM x update} \\
    \mathbf{v}_{i}^{*} &= 2\mathbf{v}_{i}^{n} - \mathbf{v}_{i}^{n-1}, \label{eq:IAEM v extrapolation}\\
    \mathbf{P}_{i}^{n+1} &= \mathbf{P}_{i}^{n} + q_i \Bigg( - \nabla \phi^{n+1} + \nabla \mathbf{A}^{n+1} \cdot \mathbf{v}_{i}^{*} \Bigg)\Delta t, \label{eq:IAEM P update} \\
    \mathbf{v}_{i}^{n+1} &\equiv \frac{c^2 \left( \mathbf{P}_{i}^{n+1} - q_{i} \mathbf{A}^{n+1}\right)}{\sqrt{ c^2\left( \mathbf{P}_{i}^{n+1} - q_{i} \mathbf{A}^{n+1}\right)^2 + \left(m_{i}c^2\right)^2}}. \label{eq:IAEM v defn}
\end{empheq}
This method, which is first-order accurate in time, begins by calculating the new particle positions using an explicit Euler advance \eqref{eq:IAEM x update}. We then obtain the charge density $\rho^{n+1}$ and an approximate current density $\tilde{\mathbf{J}}^{n+1}$ which are used to evolve the fields under the BDF discretization. The velocity is linearly extrapolated to time $t^{n+1}$ following equation \eqref{eq:IAEM v extrapolation}. Using updated fields, we compute the particle momenta according to equation \eqref{eq:IAEM P update}. Finally, the velocity at time $t^{n+1}$ is then calculated from the definition \eqref{eq:IAEM v defn}. In a comparison with the semi-implicit Euler method of \cite{Gibbon2010}, the IAEM displayed a notable reduction in the error constant and was determined to be more effective at controlling the growth in the energy.

%
%
\subsection{Properties of the BDF Wave Solver}\label{subsec:BDF-Properties}

In this section, we establish some structure preservation properties satisfied by the semi-discrete potential formulation. This includes a lemma and two theorems which provide a connection relating the semi-discrete Lorenz gauge condition \eqref{eq:semi-discrete Lorenz} to the semi-discrete continuity equation
\begin{equation}
    \label{eq:semi-discrete continuity}
    \frac{\rho^{n+1} - \rho^{n}}{\Delta t} + \nabla \cdot \textbf{J}^{n+1} = 0.
\end{equation}
We also discuss the connection to the satisfaction of Gauss' Law \eqref{eq:Gauss-E} in the semi-discrete sense.

\subsubsection{Time-consistency of the Semi-discrete Lorenz Gauge Formulation}
\label{subsubsec:time consistency for BDF-1}

We will now prove a theorem that establishes a certain time consistency property of the Lorenz gauge formulation of Maxwell's equations \eqref{eq:scalar potential eqn lorenz}-\eqref{eq:Lorenz gauge condition} when the semi-discrete BDF discretization \eqref{eq:BDF-1 semi-discrete equation} is applied to the potentials. By time-consistent, we mean that the semi-discrete system \eqref{eq:BDF-1 semi-discrete phi}-\eqref{eq:semi-discrete Lorenz} for the potentials induces an equivalence between the gauge condition and continuity equation at the semi-discrete level. In the treatment of the semi-discrete equations, we shall ignore effects of dimensional splittings and instead consider the more general inverse induced by the integral solution \eqref{eq:u layer potential solution}. Note that space is treated in a continuous manner, which means that effects due to the representation of particles on the mesh are ignored.

We show that this semi-discrete system is time-consistent in the sense of the semi-discrete Lorenz gauge. To simplify the presentation, we first prove the following lemma that connects the discrete gauge condition \eqref{eq:semi-discrete Lorenz} to the semi-discrete equations for the potentials given by equations \eqref{eq:BDF-1 semi-discrete phi} and \eqref{eq:BDF-1 semi-discrete A}.
\begin{lemma}
The semi-discrete Lorenz gauge condition \eqref{eq:semi-discrete Lorenz} satisfies the recurrence relation
\begin{equation}
    \label{eq:semi-discrete Lorenz residual update}
    \epsilon^{n+1} = \mathcal{L}^{-1} \Bigg( 2 \epsilon^{n} - \epsilon^{n-1} + \frac{\mu_0}{\alpha^2} \left( \frac{\rho^{n+1} - \rho^{n}}{\Delta t} + \nabla \cdot \mathbf{J}^{n+1} \right) \Bigg),
\end{equation}
where we have defined the semi-discrete residual 
\begin{equation}
    \label{eq:semi-discrete Lorenz residual}
    \epsilon^{n} = \frac{\phi^{n} - \phi^{n-1}}{c^{2} \Delta t} + \nabla \cdot \mathbf{A}^{n}.
\end{equation}
\label{lemma:residual equivalence}
\end{lemma}
\begin{proof}

Again, we note that the modified Helmholtz operator $\mathcal{L}$ can be formally ``inverted" with the integral solution given by equation \eqref{eq:u layer potential solution}. From this, we can calculate the terms involving the potentials in the residual \eqref{eq:semi-discrete Lorenz residual}. Proceeding, the equation for the scalar potential is found to be
\begin{equation}
    \phi^{n+1} = \mathcal{L}^{-1} \Bigg( 2 \phi^{n} - \phi^{n-1} + \frac{1}{\alpha^2 \epsilon_0} \rho^{n+1} \Bigg), \label{eq:BDF-1 semi-discrete phi update}
\end{equation}
which can be evaluated at time level $n$ to yield
\begin{equation}
    \label{eq:BDF-1 semi-discrete phi update at t_n}
    \phi^{n} = \mathcal{L}^{-1} \Bigg( 2 \phi^{n-1} - \phi^{n-2} + \frac{1}{\alpha^2 \epsilon_0} \rho^{n} \Bigg).
\end{equation}
Next, we take the divergence of $\mathbf{A}$ in equation \eqref{eq:BDF-1 semi-discrete A} and find that
\begin{equation*}
    \mathcal{L} \left( \nabla \cdot \mathbf{A}^{n+1} \right) = 2 \nabla \cdot \mathbf{A}^{n} - \nabla \cdot \mathbf{A}^{n-1} + \frac{\mu_0}{\alpha^2} \nabla \cdot \mathbf{J}^{n+1}.
\end{equation*}
Formally inverting the operator $\mathcal{L}$, we obtain the relation
\begin{equation}
    \nabla \cdot \mathbf{A}^{n+1} = \mathcal{L}^{-1} \left( 2 \nabla \cdot \mathbf{A}^{n} - \nabla \cdot \mathbf{A}^{n-1} + \frac{\mu_0}{\alpha^2} \nabla \cdot \mathbf{J}^{n+1} \right). \label{eq:BDF-1 semi-discrete div A update}
\end{equation}

With the aid of equations \eqref{eq:BDF-1 semi-discrete phi update}, \eqref{eq:BDF-1 semi-discrete phi update at t_n}, and \eqref{eq:BDF-1 semi-discrete div A update}, along with the linearity of the operator $\mathcal{L}$, a direct calculation reveals that the residual \eqref{eq:semi-discrete Lorenz residual} at time level $n+1$ is given by
\begin{equation*}
    \frac{\phi^{n+1} - \phi^{n}}{c^{2} \Delta t} + \nabla \cdot \mathbf{A}^{n+1} = \mathcal{L}^{-1} \Bigg( \frac{2\phi^{n} - 3\phi^{n-1} + \phi^{n-2}}{c^2 \Delta t} + 2 \nabla \cdot \mathbf{A}^{n} - \nabla \cdot \mathbf{A}^{n-1} + \frac{\mu_0}{\alpha^2} \left( \frac{\rho^{n+1} - \rho^{n}}{\Delta t} + \nabla \cdot \mathbf{J}^{n+1} \right) \Bigg).
\end{equation*}
Note that we have used the relation $c^2 = 1/(\mu_0 \epsilon_0)$. From these calculations, we can see that the corresponding semi-discrete continuity equation \eqref{eq:semi-discrete continuity} acts as a source for the residual \eqref{eq:semi-discrete Lorenz residual}. The remaining terms in the operand for the inverse can  also be expressed directly in terms of this semi-discrete gauge, since
\begin{align*}
    \frac{2\phi^{n} - 3\phi^{n-1} + \phi^{n-2}}{c^2 \Delta t} + 2 \nabla \cdot \mathbf{A}^{n} - \nabla \cdot \mathbf{A}^{n-1} &= 2 \left( \frac{\phi^{n} - \phi^{n-1}}{c^{2} \Delta t} + \nabla \cdot \mathbf{A}^{n} \right) - \left( \frac{\phi^{n-1} - \phi^{n-2}}{c^{2} \Delta t} + \nabla \cdot \mathbf{A}^{n-1} \right), \\
    &\equiv 2 \epsilon^{n} - \epsilon^{n-1}.
\end{align*}
This completes the proof.
\end{proof}
With the aid of Lemma \ref{lemma:residual equivalence}, we are now prepared to prove the following theorem that establishes the time-consistency of the semi-discrete system.
\begin{theorem}
The semi-discrete Lorenz gauge formulation of Maxwell's equations \eqref{eq:BDF-1 semi-discrete phi}-\eqref{eq:semi-discrete Lorenz} is time-consistent in the sense that the semi-discrete Lorenz gauge condition \eqref{eq:semi-discrete Lorenz} is satisfied at any discrete time $t^{n+1}$ if and only if the corresponding semi-discrete continuity equation \eqref{eq:semi-discrete continuity} is also satisfied.
\label{thm:consistency in time}
\end{theorem}
\begin{proof}

We use a simple inductive argument to prove both directions. In the case of the forward direction, we assume that the semi-discrete gauge condition is satisfied at any discrete time $t^{n}$ such that the residual $\epsilon^{n} \equiv 0$. Combining this with equation \eqref{eq:semi-discrete Lorenz residual update}, established by Lemma \ref{lemma:residual equivalence}, it follows that the next time level satisfies
\begin{equation*}
    0 = \mathcal{L}^{-1} \Bigg( \frac{\mu_0}{\alpha^2} \left( \frac{\rho^{1} - \rho^{0}}{\Delta t} + \nabla \cdot \mathbf{J}^{1} \right) \Bigg).
\end{equation*}
Applying the operator $\mathcal{L}$ to both sides leads to
\begin{equation*}
    \frac{\rho^{1} - \rho^{0}}{\Delta t} + \nabla \cdot \mathbf{J}^{1} = 0.
\end{equation*}
This argument can be iterated $n$ times to show that
\begin{equation*}
    \frac{\rho^{n+1} - \rho^{n}}{\Delta t} + \nabla \cdot \mathbf{J}^{n+1} = 0,
\end{equation*}
holds at any discrete time $t^{n}$, as well, which establishes the forward direction.

A similar argument can be used for the converse. Here, we show that if the semi-discrete continuity equation \eqref{eq:semi-discrete continuity} is satisfied for any time level $n$, then the residual for the discrete gauge condition also satisfies $\epsilon^{n+1} \equiv 0$. First, we assume that the initial data and starting values satisfy $\epsilon^{-1} \equiv \epsilon^{0} \equiv 0$. Appealing to equation \eqref{eq:semi-discrete Lorenz residual update} with this initial data, it is clear that after a single time step, the residual in the gauge condition satisfies
\begin{equation*}
    \epsilon^{1} = \mathcal{L}^{-1} \Bigg( 2\epsilon^{0} - \epsilon^{-1} + \frac{\mu_0}{\alpha^2} \left( \frac{\rho^{1} - \rho^{0}}{\Delta t} + \nabla \cdot \mathbf{J}^{1} \right) \Bigg) \equiv \mathcal{L}^{-1} ( 0 ).
\end{equation*}
This argument can also be iterated $n$ more times to obtain the result, which finishes the proof.
\end{proof}

Theorem \ref{thm:consistency in time} motivates, to a large extent, our choice in presenting a low-order time discretization for the fields. By repeating the calculations above with a higher-order time discretization for the fields, one can easily see that the residual for the gauge is non-vanishing, so additional modifications will be needed. The extension to higher-order time accuracy is the subject of on-going work and will be presented in a subsequent paper. While we have neglected the error introduced by the discrete maps that transfer particle data to the mesh, the results we present for the numerical experiments in section \ref{sec:5 Numerical results} suggest that these effects are not significant. We will be using Theorem \ref{thm:consistency in time} to construct maps for the charge density which satisfy the semi-discrete form of the gauge condition.

\subsubsection{Satisfying Gauss' Law}

A benefit that comes with satisfying the semi-discrete Lorenz gauge \eqref{eq:semi-discrete Lorenz} is the corresponding satisfaction of Gauss' law \eqref{eq:Gauss-E}. We summarize this with the following theorem:

\begin{theorem} \label{thm:lorenz implies gauss}
    If the semi-discrete Lorenz gauge \eqref{eq:semi-discrete Lorenz} is satisfied, then Gauss' law \eqref{eq:Gauss-E} is also satisfied.
\end{theorem}

\begin{proof}
    Consider the backward Euler discretization of the Lorenz gauge condition in time:
    \begin{equation*}
        \frac{1}{c^2}\frac{\phi^{n+1} - \phi^{n}}{\Delta t} + \nabla \cdot \textbf{A}^{n+1} = 0.
    \end{equation*}
    We know from the definition \eqref{eq:Convert potentials to EB} that
    \begin{equation*}
        \textbf{E}^{n+1} = -\nabla\phi^{n+1} - \frac{\textbf{A}^{n+1} - \textbf{A}^n}{\Delta t}.
    \end{equation*}
    Applying the divergence operator to both sides yields
    \begin{align*}
        \begin{split}
            \nabla \cdot \textbf{E}^{n+1} &= -\Delta \phi^{n+1} - \frac{\nabla\cdot\textbf{A}^{n+1} - \nabla\cdot\textbf{A}^{n}}{\Delta t} \\
            &= -\Delta\phi^{n+1} - \dfrac{\left(-\dfrac{1}{c^2}\dfrac{\phi^{n+1} - \phi^{n}}{\Delta t}\right) - \left(-\dfrac{1}{c^2}\dfrac{\phi^{n} - \phi^{n-1}}{\Delta t}\right)}{\Delta t}\\
            &= -\Delta \phi^{n+1} + \frac{1}{c^2}\frac{\phi^{n+1} - 2\phi^{n} + \phi^{n-1}}{\Delta t^2}.
        \end{split}
    \end{align*}
    From the semi-discrete update \eqref{eq:BDF-1 semi-discrete phi} for the scalar potential, the last line is equal to $\rho^{n+1}/\epsilon_{0}$, which completes the proof.
\end{proof}
It is worth noting that the above argument requires that the discrete time derivative operator applied to the Lorenz gauge be the same as the one applied to the scalar potential wave equation.

%
%

%
%

%
%
%
%
\subsection{Summary}
\label{subsec:2 Summary}

In this section, we introduced the Hamiltonian formulation of the Vlasov-Maxwell system considered in this work.  The decision to use the Hamiltonian formulation was primarily motivated by the desire to eliminate time derivatives of the vector potential resulting from a gauge formulation.  Using backwards finite difference approximations to approximate time derivatives, we introduced the semi-discrete form of the wave equations used to evolve the potentials.  The resulting algorithm is summarized in Algorithm \ref{alg:IAEM-skeleton}. We discussed key properties satisfied by the semi-discrete potential formulation that provide a useful connection between the continuity equation and Lorenz gauge condition in the semi-discrete sense. We also showed how these properties are connected to Gauss' law.  In subsequent sections, we shall use these properties to build maps that ensure the satisfaction of the continuity equation at the semi-discrete level.

\begin{algorithm}[h]
    \caption{Outline of the PIC algorithm with the improved asymmetric Euler method (IAEM)}
    Perform one time step of the PIC cycle using the improved asymmetric Euler method. 
    \label{alg:IAEM-skeleton}
    \begin{algorithmic}[1]
    \State \textbf{Given}: $(\mathbf{x}_{i}^{0}, \mathbf{P}_{i}^{0}, \mathbf{v}_{i}^{0})$, as well as the fields $ \left( \phi^{0}, \nabla \phi^{0} \right)$ and $\mathbf{A}^{0}, \nabla \mathbf{A}^{0}$
    \State Initialize $\mathbf{v}_{i}^{-1} = \mathbf{v}_{i}(-\Delta t)$ using a Taylor approximation.
    \While{stepping}
    
    \vspace{10pt}
    
    \State \label{alg:start-time-loop} Update the particle positions with $$\mathbf{x}_{i}^{n+1} = \mathbf{x}_{i}^{n} + \mathbf{v}_{i}^{n} \Delta t.$$
    
    \State Using $\textbf{x}^{n+1}$ and $\textbf{v}^{n}$, compute the current density $\textbf{J}^{n+1}$ using a bilinear mapping. \label{alg:step-current-den-naive}

    \State Using $\textbf{x}^{n+1}$, compute the charge density $\rho^{n+1}$ using a bilinear mapping. \label{alg:step-charge-den-naive}
    
    \State Compute the potentials, $\phi$ and $\textbf{A}$, at time level $t^{n+1}$ using the semi-discrete BDF method. \label{alg:step-potentials-naive}

    \State Compute the spatial derivatives of the potentials at time level $t^{n+1}$ using the semi-discrete BDF method. \label{alg:step-potential-derivatives-naive}

    \State Evaluate the Taylor corrected particle velocities $$\mathbf{v}_{i}^{*} = 2 \mathbf{v}_{i}^{n} - \mathbf{v}_{i}^{n-1}.$$
    
    \State Calculate the new generalized momentum according to $$\mathbf{P}_{i}^{n+1} = \mathbf{P}_{i}^{n} + q_i \Big( - \nabla \phi^{n+1} + \nabla \mathbf{A}^{n+1} \cdot \mathbf{v}_{i}^{*} \Big)\Delta t.$$
    
    \State Convert the new generalized momenta into new particle velocities with $$ \mathbf{v}_{i}^{n+1} =  \frac{c^2 \left( \mathbf{P}_{i}^{n+1} - q_{i} \mathbf{A}^{n+1}\right)}{\sqrt{ c^2\left( \mathbf{P}_{i}^{n+1} - q_{i} \mathbf{A}^{n+1}\right)^2 + \left(m_{i}c^2\right)^2}}. $$
    
    \State Shift the time history data and return to step \ref{alg:start-time-loop} to begin the next time step. 
    
    \EndWhile
    \end{algorithmic}
\end{algorithm}

%
%
%
%
\section{Enforcing the Semi-discrete Lorenz Gauge Condition}
\label{sec:3 exact charge maps}

In this section, we present two charge maps that conserve a semi-discrete form of the Lorenz gauge condition in an exact manner as a consequence of the theory presented in section \ref{sec:2 formulation}.  The first map we propose acquires the current density from a bilinear mapping to the mesh, which is used to solve the continuity equation to define a new charge density.  In this approach, we consider two methods for computing the numerical derivatives in the evaluation of $\nabla \cdot \textbf{J}$.  The second technique introduces an exact charge conserving map that is well suited for the multi-dimensional Green's function wave solver proposed in \cite{cheng2017asymptotic}. Although we do not implement the latter approach in this paper, we plan to explore this approach in future work.  As an alternative to charge conserving maps, we introduce a correction technique to enforce a semi-discrete Lorenz gauge condition, which corrects the scalar potential in a direct manner, resulting in a point-wise consistent set of field variables.

%
%
\subsection{A Charge Map Utilizing Numerical Derivatives}
\label{subsec:3-charge-conserving-numerical-derivatives}


We seek to create a map for the charge density that is consistent with a semi-discrete continuity equation \eqref{eq:semi-discrete continuity}. Our goal is that this map enforce the Lorenz Gauge condition. The semi-discrete properties established in Theorem \ref{thm:consistency in time} requires that the method satisfy \eqref{eq:semi-discrete continuity} point-wise on the mesh. From a practical perspective, if one goes through the derivation of Lemma \ref{lemma:residual equivalence}, an important observation is that if this Lemma is going to hold in the discrete setting, the discrete divergence operators used to compute the solution to \eqref{eq:semi-discrete continuity} must be the same discrete operator used to compute the divergence of $\mathbf{A}$. The remark in section \ref{subsec:Remark on derivatives}, which is presented later, provides further clarification on this matter.     

To create our map for the charge density, we will solve for $\rho^{n+1}$ using an update of the semi-discrete continuity equation \eqref{eq:semi-discrete continuity}. The current density $\mathbf{J}^{n+1}$  used in this update is obtained from the particles using a bilinear map. To compute the discrete divergence operator, we will utilize either the fast Fourier transform (FFT) or the sixth-order centered finite difference method (FD6). Additionally, to make the point-wise map conservative, we introduce a Lagrange multiplier, which we now discuss. In the construction, we make the following observation: Starting from the non-corrected point-wise update
$$\rho_{i,j}^{n+1} = \rho_{i,j}^{n} - \Delta t \nabla \cdot \mathbf{J}_{i,j}^{n+1},$$ 
after summing over the periodic domain, we obtain
\begin{equation}
    \sum_{i,j}\rho_{i,j}^{n+1} = \sum_{i,j}\rho_{i,j}^{n} - \Delta t\sum_{i,j}\left(\nabla \cdot \mathbf{J}^{n+1}\right)_{i,j}.
\end{equation}
To ensure charge conservation, we require $$\sum_{i,j}\rho^{n+1} = \sum_{i,j}\rho^{n}, $$ which further implies the condition
$$\sum_{ij}\left(\nabla \cdot \mathbf{J}^{n+1}\right)_{ij} = 0.$$ 
We introduce the Lagrange multiplier, $\gamma$, to enforce this condition. For 2D, we define
\begin{equation*}
    \gamma \coloneqq -\frac{1}{N_xN_y}\sum_{i,j}{\left(\nabla\cdot \mathbf{J}\right)_{i,j}}, \quad \textbf{F} \coloneqq \frac{1}{2}\gamma\textbf{x},
\end{equation*}
and let $\mathbf{J}^* \coloneqq \mathbf{J} + \textbf{F}$ be the adjusted current density. In 3D, the $\frac{1}{2}$ would become a $\frac{1}{3}$. Correspondingly,
\begin{align}
    \begin{split}
        \sum_{i,j}\left(\nabla \cdot \mathbf{J}^{n+1,*}\right)_{i,j} &= \sum_{i,j}\left(\nabla \cdot \mathbf{J}^{n+1} + \nabla \cdot \textbf{F}\right)_{i,j} \\
        &= \sum_{i,j}\left(\nabla\cdot\mathbf{J}^{n+1} - \frac{1}{N_xN_y}\sum_{l,k}\left(\nabla \cdot \mathbf{J}^{n}\right)_{l,k}\right) \\
        &= \sum_{i,j}\left(\nabla \cdot \mathbf{J}^{n+1}\right)_{i,j} - \sum_{i,j}\left(\nabla\cdot\mathbf{J}^{n+1}\right)_{i,j} = 0.
    \end{split}
\end{align}
We now use $\rho^{n+1}$ and $\mathbf{J}^{n+1,*}$ in the update of $\phi^{n+1}$ and $\mathbf{A}^{n+1}$.

As noted earlier, in the above expressions, the 
divergence $\nabla \cdot \mathbf{J}_{i,j}^{n+1,*}$ is computed using either an FFT or a sixth-order finite differnce. For the FFT, we have the identity
\begin{equation*}
    \partial_{x}u = \mathcal{F}^{-1}_{x}\left[ik_x\mathcal{F}\left[u\right]\right], 
\end{equation*}
where $k_{x}$ is the wavenumber. Derivatives in $y$ can be calculated using identical formulas. For FD6, we use the following centered finite-difference stencil
\begin{equation*}
    \partial_{x} u_{i,j} = -\frac{1}{60}u_{i-3,j} + \frac{3}{20}u_{i-2,j} -\frac{3}{4}u_{i-1,j} + \frac{3}{4}u_{i+1,j} -\frac{3}{20}u_{i+2,j} + \frac{1}{60}u_{i+3,j},
\end{equation*}
with an analogous set of coefficients for the $y$ derivative. It should be mentioned that in numerical experiments, the schemes preserve the gauge condition to machine precision, with and without the Lagrange multiplier, using either the FFT or FD6 in the update of $\rho^{n+1}$ (see section \ref{sec:5 Numerical results}).

In summary, Algorithm \ref{alg:IAEM-skeleton} is modified in two ways. Line \ref{alg:step-charge-den-naive} is changed by eliminating the bilinear interpolation of particles to obtain $\rho^{n+1}$. Instead, $\rho^{n+1}$ is computed from the semi-discrete continuity equation \eqref{eq:semi-discrete continuity}, with line \ref{alg:step-current-den-naive} supplying the value of $\textbf{J}^{n+1}$ for whatever discrete derivative operator has been decided. Line \ref{alg:step-potential-derivatives-naive} computes the derivatives acquired from line \ref{alg:step-potentials-naive}, which remains unchanged, and it does so using the same discrete derivative operator, for reasons we now move on to explain.

\subsubsection{A Remark on the Derivatives}
\label{subsec:Remark on derivatives}

It is worth commenting upon the nature of the way the $\nabla \textbf{A}$ and $\nabla\phi$ values are obtained in the context of the IAEM method (section \ref{subsec:IAEM}). As stated above, the BDF method does have the ability to compute analytically the derivatives of $\textbf{A}^{n+1}$ and $\phi^{n+1}$ from the values of their previous timesteps, combined with the sources $\textbf{J}$ and $\rho$, respectively. There is a difficulty in this approach, however. The proof of Theorem \ref{thm:consistency in time}, which connects the gauge condition to the continuity equation, is predicated upon Lemma \ref{lemma:residual equivalence}, the proof of which assumes the derivative operator used on $\textbf{A}$ and $\textbf{J}$ are the same. In order to assert \eqref{eq:semi-discrete Lorenz residual update}, the following relationship must hold:
\begin{align*}
    \textbf{A}^{n+1} = \mathcal{L}^{-1}\left[2\textbf{A}^{n} - \textbf{A}^{n-1} + \frac{\mu_0}{\alpha^2}\textbf{J}^{n+1}\right] \implies \nabla \cdot \textbf{A}^{n+1} = \mathcal{L}^{-1}\left[2\nabla \cdot \textbf{A}^{n} - \nabla \cdot \textbf{A}^{n-1} + \frac{\mu_0}{\alpha^2}\nabla \cdot \textbf{J}^{n+1}\right].
\end{align*}
Of course, in the continuous case, this is trivial. However, should we, for example, update $\rho$ based on a discrete Fourier transform of $\nabla \cdot \textbf{J}$, then this theorem does not hold unless the spatial derivative applied to $\textbf{A}$ is likewise obtained using a discrete Fourier transform. In the fully discrete setting, Lemma \ref{lemma:residual equivalence} requires that the discrete derivative operators used to compute $\nabla \cdot \mathbf{A}$ and $\nabla \cdot \mathbf{J}$ have the same form.

%
%
%
%
%
\subsection{A Charge Map for the Multi-dimensional Boundary Integral Solution}
\label{subsec:3-charge-conserving-boundary-integral-solution}

The map used to construct $\rho^{n+1}$ in the previous section has a natural extension to boundary integral formulations that evolve the potentials using multi-dimensional Green's functions. The solutions for the potentials in terms of the boundary integral representation \eqref{eq:u layer potential solution} leads to the equations
\begin{align}
    \phi^{n+1}(\mathbf{x}) &= \int_{\Omega} G(\mathbf{x},\mathbf{y}) \left( 2 \phi^{n} - \phi^{n-1} + \frac{1}{\alpha^2 \epsilon_0} \rho^{n+1} \right)(\mathbf{y}) \, dV_{\mathbf{y}} + \int_{\partial \Omega} \left(  \sigma_{\phi}(\mathbf{y}) G(\mathbf{x},\mathbf{y}) + \gamma_{\phi}(\mathbf{y})\frac{\partial G}{\partial \mathbf{n}} \right) \, dS_{\mathbf{y}}, \label{eq:phi BIE solution} \\
    \mathbf{A}^{n+1}(\mathbf{x}) &= \int_{\Omega} G(\mathbf{x},\mathbf{y}) \left( 2 \mathbf{A}^{n} - \mathbf{A}^{n-1} + \frac{\mu_0}{\alpha^2} \mathbf{J}^{n+1} \right)(\mathbf{y}) \, dV_{\mathbf{y}} + \int_{\partial \Omega} \left(  \sigma_{\mathbf{A}}(\mathbf{y}) G(\mathbf{x},\mathbf{y}) + \gamma_{\mathbf{A}}(\mathbf{y})\frac{\partial G}{\partial \mathbf{n}} \right) \, dS_{\mathbf{y}}. \label{eq:A BIE solution}
\end{align}
When particles are introduced, maps to the mesh for $\rho^{n+1}$ and $\mathbf{J}^{n+1}$ generally do not satisfy the semi-discrete continuity equation \eqref{eq:semi-discrete continuity}, which, by Lemma \ref{lemma:residual equivalence}, also means that the semi-discrete Lorenz gauge condition will be violated; however, if the map for $\rho^{n+1}$ can be connected to the map for $\mathbf{J}^{n+1}$, then this constraint will indeed be satisfied. Following the observation of the previous section, we shall construct this map by solving equation \eqref{eq:semi-discrete continuity} for the charge at the new time level according to
\begin{equation*}
    \rho^{n+1} = \rho^{n} - \Delta t ~ \nabla \cdot \mathbf{J}^{n+1}.
\end{equation*}
This discretization, which is first-order in time, is consistent with the BDF field solver considered in this work. An obvious extension to higher-order is to include additional stencil points for the approximation of time derivative, which will be considered at a later time. Note that at the initial time $t^{0} = 0$, the particle data is known, and therefore, the charge and current densities $\rho^{0}$ and $\mathbf{J}^{0}$ are also known.  

The primary objective is to define the volume integral term
\begin{equation}
    \label{eq:multi-D map starting point}
    \int_{\Omega} G(\mathbf{x},\mathbf{y}) \rho^{n+1} (\mathbf{y}) \, dV_{\mathbf{y}} = \int_{\Omega} G(\mathbf{x},\mathbf{y}) \rho^{n} (\mathbf{y}) \, dV_{\mathbf{y}} -  \Delta t \int_{\Omega} G(\mathbf{x},\mathbf{y}) \nabla_{y} \cdot \mathbf{J}^{n+1} (\mathbf{y}) \, dV_{\mathbf{y}}.
\end{equation}
so that it is compatible with the low regularity charge and current data produced by particle simulations. In PIC, we construct $\mathbf{J}^{n+1}$ by scattering the current density of the particles to the mesh using particle shape functions that typically have low regularity. In most cases, especially in bounded domains, these maps are either piece-wise constant or linear functions, or combinations of these, so we would like to avoid taking derivatives of this term. With the aid of the divergence theorem, we can move the derivatives from the less regular particle data to the Green's function $G$ whose derivatives can be evaluated analytically. This allows us to write
\begin{equation*}
\int_{\Omega} G(\mathbf{x},\mathbf{y}) \nabla_{y} \cdot \mathbf{J}^{n+1} (\mathbf{y}) \, dV_{\mathbf{y}} = \int_{\partial\Omega} G(\mathbf{x},\mathbf{y}) ~\mathbf{J}^{n+1} (\mathbf{y}) \cdot \mathbf{n} \left(\mathbf{y}\right) \, dS_{\mathbf{y}} -\int_{\Omega} \nabla_{y} G(\mathbf{x},\mathbf{y}) \cdot \mathbf{J}^{n+1} (\mathbf{y}) \, dV_{\mathbf{y}},
\end{equation*}
which can be combined with equation \eqref{eq:multi-D map starting point}. This yields the map
\begin{align}
    \int_{\Omega} G(\mathbf{x},\mathbf{y}) \rho^{n+1}(\mathbf{y}) \, dV_{\mathbf{y}} &= \int_{\Omega} G(\mathbf{x},\mathbf{y}) \rho^{n}(\mathbf{y}) \, dV_{\mathbf{y}} \nonumber \\
    &- \Delta t \left[\int_{\partial\Omega} G(\mathbf{x},\mathbf{y}) ~\mathbf{J}^{n+1} (\mathbf{y}) \cdot \mathbf{n} \left(\mathbf{y}\right) \, dS_{\mathbf{y}} -\int_{\Omega} \nabla_{y} G(\mathbf{x},\mathbf{y}) \cdot \mathbf{J}^{n+1} (\mathbf{y}) \, dV_{\mathbf{y}}\right], \label{eq:multi-D charge map BIE}
\end{align}
which, by Theorem \ref{thm:consistency in time}, will be charge conserving and enforce the Lorenz gauge condition in the semi-discrete sense.

The charge map given by equation \eqref{eq:multi-D charge map BIE} can then be substituted into the boundary integral solution \eqref{eq:phi BIE solution} to give the update
\begin{align}
    \label{eq:phi BIE with map}
    \phi^{n+1}(\mathbf{x}) &= \int_{\Omega} G(\mathbf{x},\mathbf{y}) \left( 2 \phi^{n} - \phi^{n-1} + \frac{1}{\alpha^2 \epsilon_0}  \rho^{n} \right)(\mathbf{y}) \, dV_{\mathbf{y}} + \int_{\partial \Omega} \left(  \sigma_{\phi}(\mathbf{y}) G(\mathbf{x},\mathbf{y}) + \gamma_{\phi}(\mathbf{y})\frac{\partial G}{\partial \mathbf{n}} \right) \, dS_{\mathbf{y}} \nonumber \\ 
    &+  \frac{\Delta t}{\alpha^2 \epsilon_0} \left[ \int_{\Omega} \nabla_{y} G(\mathbf{x},\mathbf{y}) \cdot \mathbf{J}^{n+1} (\mathbf{y}) \, dV_{\mathbf{y}} - \int_{\partial\Omega} G(\mathbf{x},\mathbf{y}) ~\mathbf{J}^{n+1} (\mathbf{y}) \cdot \mathbf{n} \left(\mathbf{y}\right) \, dS_{\mathbf{y}} \right].
\end{align}
Expressions for the gradients of the scalar potential $\nabla_{x} \phi^{n+1}(\mathbf{x})$ can then be obtained by differentiating equation \eqref{eq:phi BIE with map}. Since the derivatives are in $\mathbf{x}$, they are transferred directly onto analytical functions. We remark that the modifications required to achieve compatibility with dimensionally-split solvers is non-trivial, since the divergence of the current density effectively couples the boundary conditions between different directions. However, the map \eqref{eq:phi BIE with map} would be well-suited for solvers that use the full boundary integral solution \cite{cheng2017asymptotic}.

Whereas in the previous approach we explicitly computed $\rho^{n+1}$ from $\rho^{n}$ and $\textbf{J}^{n+1}$ using \eqref{eq:semi-discrete continuity}, thus modifying line \ref{alg:step-charge-den-naive} and correspondingly \ref{alg:step-potential-derivatives-naive}, this approach circumnavigates this by instead replacing computing $\rho^{n+1}$ entirely, cutting out line \ref{alg:step-charge-den-naive}. We then modify lines \ref{alg:step-potentials-naive} and \ref{alg:step-potential-derivatives-naive} in Algorithm \ref{alg:IAEM-skeleton} both by computing $\phi^{n+1}$ using a combination of $\rho^{n}, \phi^{n}, \phi^{n-1},$ and $\textbf{J}^{n+1}$ in the BDF method, and likewise with the spatial derivatives. It should be noted that the computation of $\textbf{A}^{n+1}$ remains unchanged.

%
%
%
%
\subsection{Correcting $\phi$ Using the Gauge Condition}
\label{subsec:3-corrected_phi}
As an alternative approach to solving the continuity equation, we assume that a bilinear map of $\mathbf{J}^{n+1}$ is exact and use the gauge condition its self to create a correction to $\phi^{n+1}$. To do this, we first use a bilinear map to construct both the \lq\lq exact''  $\mathbf{J}^{n+1}$ and an approximate $\rho^{n+1}$ called $\rho_{A}^{n+1}$.  These are used to construct the \lq\lq exact'' vector potential $\mathbf{A}^{n+1}$ and the approximate scalar potential $\phi_A^{n+1}$. The correction to the scalar potential, $\phi_{C}^{n+1}$, that we propose makes use of the analytical derivatives \cite{christlieb2024pic} in the evaluation of the divergence of $\mathbf{A}^{n+1}$. To start, we update the vector potential and scalar potential using bilinear maps for the particle data, which gives 
\begin{align*}
    \mathbf{A}^{n+1}\left(x,y\right) &=  \mathcal{L}_{y}^{-1} \mathcal{L}_{x}^{-1} \left[ 2 \mathbf{A}^{n} - \mathbf{A}^{n-1} + \frac{1}{\alpha^{2}} \mu_0 \mathbf{J}^{n+1} \right]\left(x,y\right),  \\
    \phi^{n+1}_A\left(x,y\right) &=  \mathcal{L}_{y}^{-1} \mathcal{L}_{x}^{-1} \left[ 2 \phi^{n} - \phi^{n-1} + \frac{1}{\alpha^{2} \epsilon_{0}} \rho_{A}^{n+1} \right]\left(x,y\right),   
\end{align*}
where  $\phi^n$ and $\phi^{n-1}$ are the corrected values such that $\mathbf{A}^n$, $\mathbf{A}^{n-1}$, $\phi^n$ and $\phi^{n-1}$ satisfy the gauge condition. We decompose the scalar potential as
\begin{equation}
    \label{eq:gauge correcting decomposition}
    \phi^{n+1} := \phi^{n+1}_{A} + \phi^{n+1}_{C},
\end{equation}
where $\phi^{n+1}_{A}$ is the approximate potential and $\phi^{n+1}_{C}$ is its correction. We also assume that $\mathbf{A}^{n+1}$ and $\phi^{n+1}$ satisfy the gauge condition in the sense that
\begin{equation*}
\frac{1}{c^2} \frac{\phi^{n+1}-\phi^{n}}{\Delta t} + \nabla \cdot \mathbf{A}^{n+1} =0.
\end{equation*}
Since this is designed to be a correction method, the discrete divergence can be evaluated using a variety of methods, including the analytical differentiation techniques developed in our previous work \cite{christlieb2024pic}. Rearranging the expression for the gauge error, we obtain the correction
\begin{equation*}
\phi^{n+1}_C =\phi^{n}  -\phi^{n+1}_A - c^2 \Delta t \left ( \partial_{x} A^{(1), n+1} + \partial_{y} A^{(2), n+1} \right ).
\end{equation*}
Once the correction $\phi^{n+1}_C$ is obtained, we then form the total potential according to the definition \eqref{eq:gauge correcting decomposition} and compute its numerical derivatives for the particle advance used in the algorithm (no significant difference was seen between FFT and FD6).  The use of a numerical derivative in the particle advance is because the terms that make up $\phi^{n+1}_C$ do not sit under the operator $\mathcal{L}^{-1}(\cdot)$.  

This is the least intrusive of all our methods, being accomplished by simply adding a line after \ref{alg:step-potential-derivatives-naive}, in which we compute the correction to $\phi$ and its numerical derivatives.  As will be seen in the numerical results, by construction, this method satisfies the gauge error to machine precision.  Also, it is worth noting that if one needs the correct $\rho$, one can set up a iterative method that can compute $\rho^{n+1}$ from $\phi^{n+1}$.  By Lemma \ref{lemma:residual equivalence}, $\rho^{n+1}$ and $\mathbf{J}^{n+1}$ will satisfy \eqref{eq:semi-discrete continuity} because $\phi^{n+1}$ and $\mathbf{A}^{n+1}$ satisfy \eqref{eq:semi-discrete Lorenz}.  Doing so would allow the use of analytical derivatives for  $\phi^{n+1}$, as in \cite{christlieb2024pic}, in the particle update and would be useful in exploring problems with geometry. This will be explored in future work.   


%
%
%
%
\subsection{Summary}
\label{subsec:3 Summary}

In this section, we presented gauge-conserving maps that use Theorem \ref{thm:consistency in time} in a direct manner to enforce the Lorenz gauge in the semi-discrete sense.  The first map computes the current density $\textbf{J}^{n+1}$ using a standard bilinear map, then computes $\rho^{n+1}$ as the solution of a semi-discrete continuity equation.  A Lagrange multiplier was also used to make certain that the continuity equation was enforced globally, though numerical experiments suggest this may not be necessary.  The second map we proposed is based on the multi-dimensional Green's function and is well-suited for solvers that leverage a boundary integral formulation, and will be investigated in future work.  These maps, which are based on semi-discrete theory, do not account for errors introduced by the spatial discretization.  Additionally, we introduced a third approach which identifies a correction to $\phi$ based on the discrete gauge condition.  The changes made to Algorithm \ref{alg:IAEM-skeleton} for all three methods are summarized in Table \ref{tab:algorithm modifications}.  In the next section, we investigate the capabilities of the first and third methods in problems defined on periodic domains.  We note that in all cases, developing appropriate boundary conditions for bounded domains for inflow and outflow of charge that is consistent with gauge preservation is the next important step, and this will be the subject of future work.  

\begin{table}[!htb]
    \centering
    \def\arraystretch{1.2}
    \begin{tabular}{ | c || c | }
        \hline
        \textbf{Method}  & \textbf{Changes} \\
        \hline
        Numerical Derivatives (FFT and FD6) & Lines \ref{alg:step-current-den-naive} and \ref{alg:step-potential-derivatives-naive} \\
        Multi-dimensional Boundary Integral & Lines \cancel{\ref{alg:step-charge-den-naive}},  \ref{alg:step-potentials-naive}, and \ref{alg:step-potential-derivatives-naive} \\
        Gauge correction for $\phi$ & Add line to correct $\phi$ after \ref{alg:step-potential-derivatives-naive}. \\
        \hline
    \end{tabular}
    \caption{The three methods we introduce and how they change Algorithm \ref{alg:IAEM-skeleton}.}
    \label{tab:algorithm modifications}
\end{table}

%
%
%
%
\section{Numerical Experiments}
\label{sec:5 Numerical results}

This section contains the numerical results for the PIC method using the exact and approximate methods to enforce the Lorenz gauge condition. We consider two different test problems with dynamic and steady state qualities. We first consider the Weibel instability, which is a streaming instability that occurs in a periodic domain. The second example we consider is a simulation of a non-relativistic drifting cloud of electrons against a stationary cluster of ions. In each example, we compare the performance of the different methods by inspecting the fully discrete Lorenz gauge condition and monitoring its behavior as a function of time. We conclude the section by summarizing the key results of the experiments.

%
%
\subsection{Relativistic Weibel Instability}
\label{subsec:weibel instability}

The Weibel instability \cite{Weibel1959} is a type of streaming instability resulting from an anistropic distribution of momenta, which is prevalent in many applications of high-energy-density physics including astrophysical plasmas \cite{WeibelAstro2021} and fusion applications \cite{FluidKineticICFsims2005}. In such circumstances, the momenta in different directions can vary by several orders of magnitude. Initially, the strong currents generated from the momenta create filament-like structures that eventually interact due to the growth in the magnetic field. Over time, the magnetic field can become quite turbulent, and the currents self-organize into larger connected networks. During this self-organization phase, there is an energy conversion mechanism that transfers the kinetic energy from the plasma into the magnetic field, which attempts to make the distribution of momenta more isotropic. The resulting instability creates the formation of magnetic islands and other structures which store massive amounts of energy. In such highly turbulent regions, this can lead to the emergence of other plasma phenomena such as magnetic reconnection, in which energy is released from the fields back into the plasma \cite{FonsecaWeibel03}.

The setup for this test problem can be described as follows. The domain for the problem is a periodic box defined on $[-L_{x}/2, L_{x}/2] \times [-L_{y}/2, L_{y}/2]$ in units of the electron skin depth $c/\omega_{pe}$. The particular values of $L_{x}$ and $L_{y}$ are carefully chosen using the dispersion relation and will be specified later. Here, $c$ is the speed of light and $\omega_{pe}$ is the electron angular plasma frequency. The system being modeled consists of two groups of infinite counter-propagating sheets of electrons and a uniform background of stationary ions. The electrons in each group are prescribed momenta from the ``waterbag" distribution
\begin{equation}
    \label{eq:Weibel IC}
    f(\mathbf{P}) = \frac{1}{2 P_{\myparallel}}\delta \left(P^{(1)} - P_{\perp}\right) \left[  \Theta \left( P^{(2)} - P_{\myparallel}\right) - \Theta\left( P^{(2)} + P_{\myparallel}\right) \right],
\end{equation}
where $\Theta(x)$ is the unit step function, and we choose $P_{\perp} > P_{\myparallel} > 0$ to induce an instability. Along the $x$ direction, electrons are prescribed only a drift corresponding to $P_{\perp}$, and are linearly spaced in the interval $[-P_{\myparallel}, P_{\myparallel})$ in the $y$ direction. We make the problem charge neutral by setting the positions of the electrons to be equal to the ions. Electrons belonging to the first group are initialized according to the distribution \eqref{eq:Weibel IC}, and those belonging to the second group are defined to be a mirror image (in momenta) of those in the first. This creates a return current that guarantees current neutrality in the initial data. For simplicity, we shall assume that the sheets have the same number density $\bar{n}$, but this is not necessary. In fact, some earlier studies explored the structure of the instability for interacting streams with different densities and drift velocities \cite{LeeFilamentation1973}. The particular values of the plasma parameters used in this experiment are summarized in Table \ref{tab:weibel plasma parameters}. We remark that during the initialization phase of our implementation, we prescribe the velocities of the particles and then convert them to conjugate momenta using $\mathbf{P}_{i} = \gamma m_{i} \mathbf{v}_{i}$, since $\mathbf{A} = 0$ at the initial time. The electron velocities used to construct the distribution \eqref{eq:Weibel IC} are equivalently given in units of $c$ by $v_{\perp} = 1/2$ and $v_{\myparallel} = 1/100$. When these values are converted to conjugate momenta, we obtain $P_{\perp} \approx 5.773888 \times 10^{-1}$ and $P_{\myparallel} \approx 1.154778 \times 10^{-2}$, respectively, which are given in units of $m_{e}c$. Lastly, we note that the normalized permittivity and normalized permeability are given by $\sigma_{1} = \sigma_{2} = 1$, and the normalized speed of light is $\kappa = 1$.

\begin{table}[!ht]
    \centering
    \def\arraystretch{1.2}
    \begin{tabular}{ | c || c | }
        \hline
        \textbf{Parameter}  & \textbf{Value} \\
        \hline
        Average number density ($\bar{n}$) [m$^{-3}$] & $1.0\times 10^{10}$ \\
        Average electron temperature ($\bar{T}$) [K] & $1.0\times 10^{4}$ \\
        Electron angular plasma period ($\omega_{pe}^{-1}$) [s/rad] & $1.772688\times 10^{-7}$ \\
        Electron skin depth ($c/\omega_{pe})$ [m] & $5.314386\times 10^{1}$ \\
        Electron drift velocity in $x$ ($v_{\perp}$) [m/s] & $c/2$ \\
        Maximum electron velocity in $y$ ($v_{\myparallel}$) [m/s] & $c/100$ \\
        \hline
    \end{tabular}
    \caption{Plasma parameters used in the simulation of the Weibel instability. All simulation particles are prescribed a drift velocity corresponding to $v_{\perp}$ in the $x$ direction while the $y$ component of their velocities are sampled from a uniform distribution scaled to the interval $[-v_{\myparallel}, v_{\myparallel})$.}
    \label{tab:weibel plasma parameters}
\end{table}

In other papers, which numerically simulate the Weibel instability, e.g., \cite{LeeFilamentation1973,Chen-ImplicitPIC-Darwin2014, MorseWeibel1971}, the particle velocities are obtained using specialized sampling methods that mitigate the effects of noise in the starting distribution. As pointed out in \cite{MorseWeibel1971}, this can lead to electrostatic effects, such as Landau damping, which will alter the growth of the instability. This motivates the consideration of the distribution \eqref{eq:Weibel IC}, as it is easy to compare with theory and control effects from sampling noise. For such distributions, the growth rate in the magnetic field can be calculated from the dispersion relation \cite{yoon1987exact}
\begin{equation*}
    1 - \frac{c^2 k_{\myparallel}^{2}}{\omega^{2}} - \frac{\omega_{pe}^{2}}{\omega^{2} \gamma} \left[ G(\beta_{\myparallel}) + \frac{\beta_{\perp}^{2}}{2 \left(1 - \beta_{\myparallel}^{2}\right)} \left( \frac{c^2 k_{\myparallel}^{2} - \omega^{2}}{\omega^{2} - c^{2} k_{\myparallel}^{2} \beta_{\myparallel}^{2} } \right) \right] = 0.
\end{equation*}
Here, $k_{\myparallel}$ is the wavenumber in the $y$ direction, $\omega_{pe}$ is the electron plasma frequency and
\begin{equation*}
    \gamma := \sqrt{1 + \frac{P_{\perp}^{2}}{m_{e}^{2} c^{2}} + \frac{P_{\myparallel}^{2}}{m_{e}^{2} c^{2}} }, \quad \beta_{\perp} := \frac{P_{\perp}}{\gamma m_{e} c}, \quad \beta_{\myparallel} := \frac{P_{\myparallel}}{\gamma m_{e} c}, \quad G(\beta_{\myparallel}) := \frac{1}{2 \beta_{\myparallel}} \ln \left( \frac{1 + \beta_{\myparallel}}{1 - \beta_{\myparallel}} \right),
\end{equation*}
with $m_{e}$ being the electron mass. In particular, Yoon and Davidson \cite{yoon1987exact} showed that an instability occurs if and only if the condition
\begin{equation*}
    \frac{\beta_{\perp}^{2}}{2 \beta_{\myparallel}^{2}} > \left( 1 - \beta_{\myparallel} \right) G(\beta_{\myparallel}),
\end{equation*}
is satisfied, and that its corresponding growth rate can be directly calculated using the equation
\begin{equation}
    \label{eq:Yoon Weibel growth rate}
    \text{Im}(\omega) = \frac{1}{\sqrt{2}} \left[ \left( W_{1}^{2} + W_{2} \right)^{1/2} - W_{1} \right]^{1/2},
\end{equation}
where we have introduced the definitions
\begin{equation*}
    W_{1} := c^{2} k_{\myparallel}^{2} \beta_{\myparallel}^{2} + \frac{\omega_{pe}^{2} \beta_{\perp}^{2}}{2 \gamma \beta_{\myparallel}^{2}} - c^{2} \left( k_{0}^{2} - k_{\myparallel}^{2} \right), \quad
    W_{2} := 4c^{4} k_{\myparallel}^{2} \beta_{\myparallel}^{2} \left( k_{0}^{2} - k_{\myparallel}^{2} \right),  \quad
    k_{0} := \frac{\omega_{pe}^{2}}{\gamma c^{2}} \left( \frac{\beta_{\perp}^{2}}{2 \beta_{\myparallel}^{2} \left( 1 - \beta_{\myparallel}^{2} \right)} - G(\beta_{\myparallel}) \right). 
\end{equation*}
Finally, we remark that the range of admissible values of $k_{\myparallel}$ for the instability is $0 < k_{\myparallel}^{2} < k_{0}^{2}$.

\begin{figure}[!htb]
    \centering
    \includegraphics[clip, trim={0cm, 0cm, 0cm, 0cm}, scale=0.25]{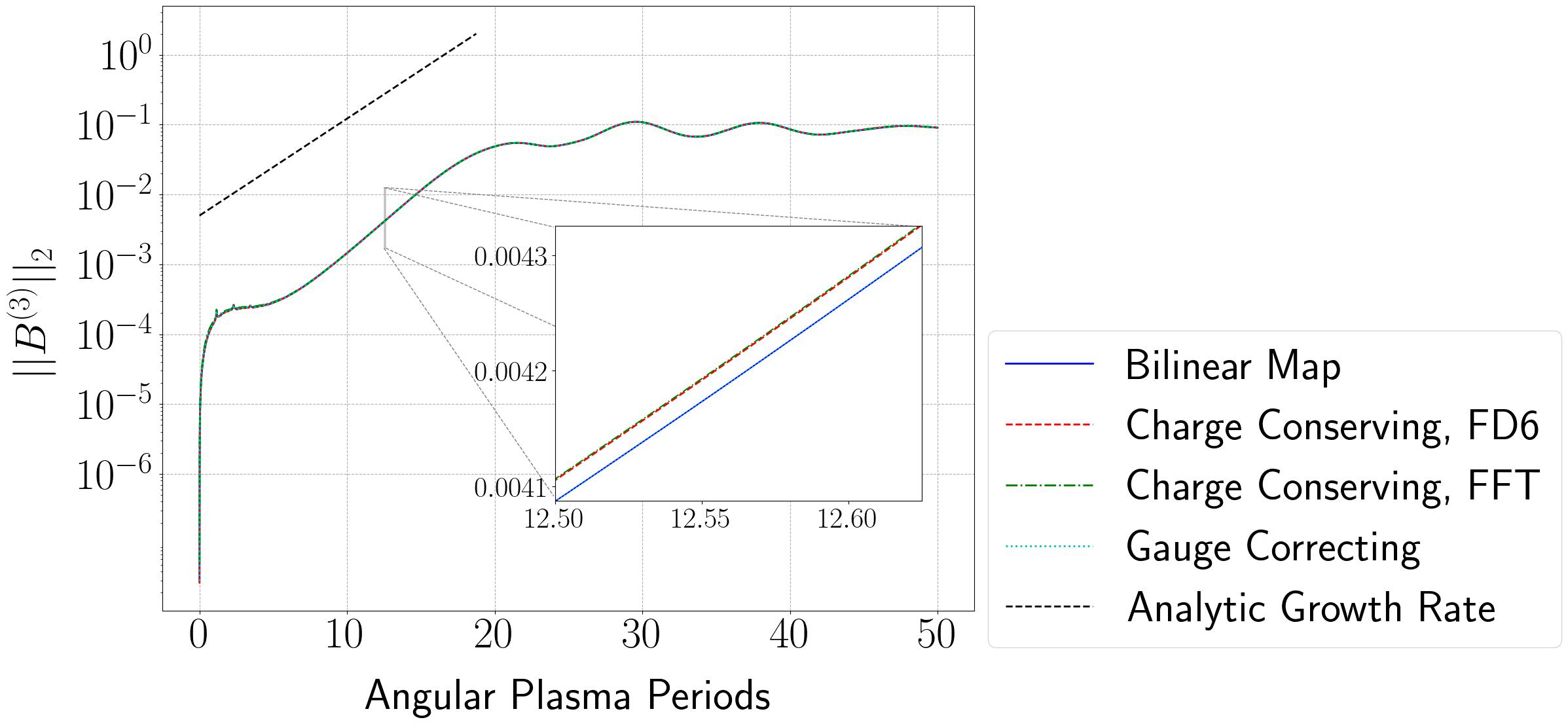}
    \caption{Growth in the magnetic field energy for the Weibel instability. We compare the growth rate in the $\ell_{2}$-norm of the magnetic field $B^{(3)}$ for different methods against an analytical growth rate predicted from linear response theory \cite{yoon1987exact}. The analytical growth rate for this configuration is determined to be $\text{Im}(\omega) \approx 0.319734$. We observe good agreement with the theoretically predicted growth rate for each of the methods.}
    \label{fig:Weibel Magnetic Magnitude}
\end{figure}

\begin{figure}[!htb]
    \centering
    \includegraphics[clip, trim={0cm, 0cm, 0cm, 0cm}, scale=0.25]{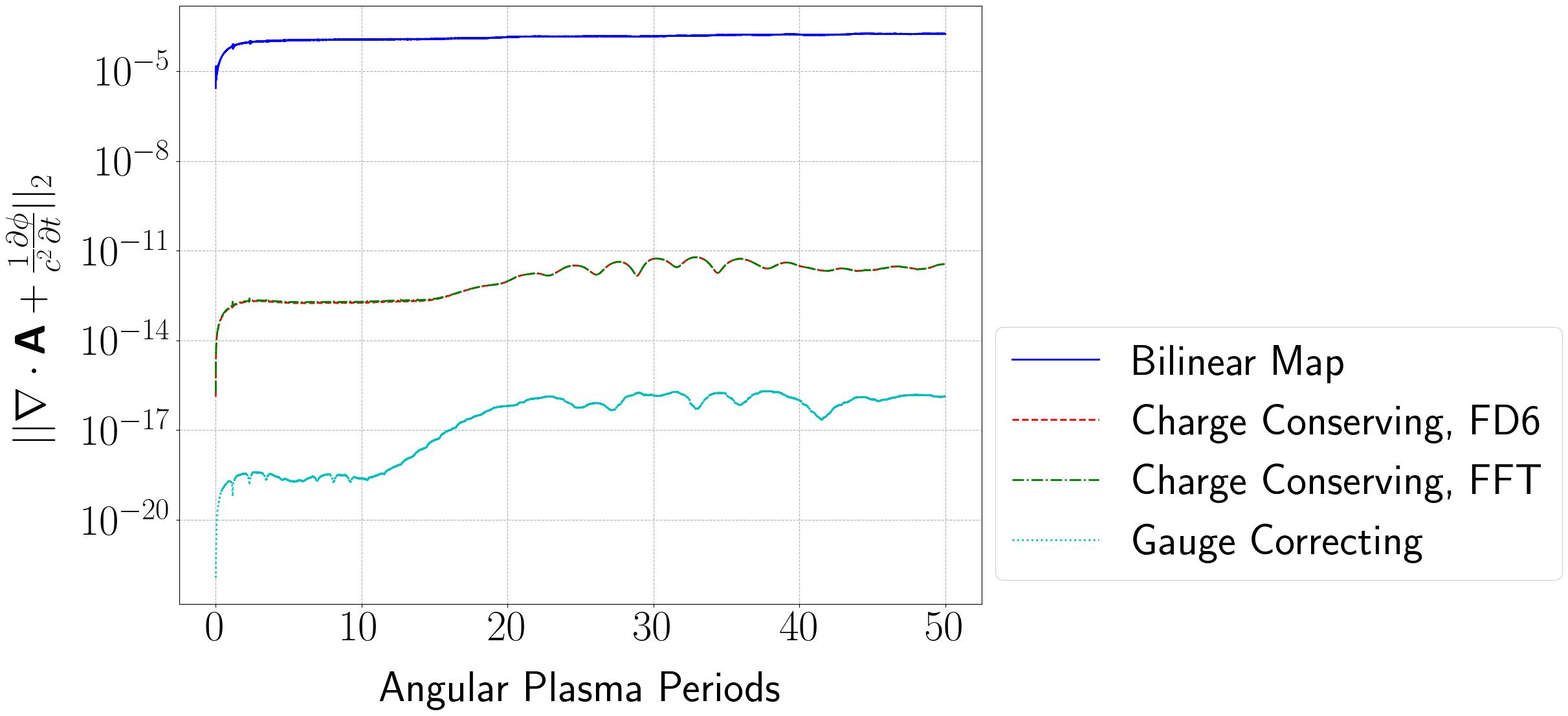}
    \caption{The gauge error for different methods applied to the Weibel instability test problem. Naive interpolation introduces significant gauge error, whereas the maps based on the continuity equation or gauge correction result in substantial improvement. It should be noted that while FFT and FD6 appear identical, they differ by $\mathcal{O}(10^{-14})$. The gauge correcting approach produced the smallest gauge error among the methods we considered.}
    \label{fig:Weibel Gauge Error}
\end{figure}

Using the parameters derived from Table \ref{tab:weibel plasma parameters} in an appropriately rescaled equation of the growth rate \eqref{eq:Yoon Weibel growth rate}, we find the value of $k_{\myparallel}$ that achieves the maximum growth is $k_{\myparallel} \approx 5.445191$. Based on this value, we scale the size of the (normalized) simulation domain so that $L_{x} = L_{y} = 1.153896$ in units of $c/\omega_{pe}$. The corresponding theoretical growth rate in the magnetic field associated with this problem is $ \text{Im}(\omega) \approx 0.319734$. We ran the simulation to a final time of $T = 50$ in units of $\omega_{pe}^{-1}$ with a $128 \times 128$ mesh. The total number of simulation particles used in the experiment was set to $1 \times 10^{6}$ and was split equally between ions and electrons. Each electron group, therefore, consists of $2.5 \times 10^{5}$ particles, and the number of particles per cell was approximately $61$. In Figure \ref{fig:Weibel Magnetic Magnitude}, we compare the growth rates in the $\ell_{2}$-norm of the magnetic field $B^{(3)}$ obtained with different methods. We find that each of the methods, including the naive approach, show good agreement with the theoretical growth rate predicted from equation \eqref{eq:Yoon Weibel growth rate}. Additionally in Figure \ref{fig:Weibel Gauge Error}, we plot the error in the gauge condition obtained with each of the methods. The naive approach, which uses bilinear maps for the charge and current densities introduces significant gauge error and increases over time. In contrast, the maps which use either the continuity equation or the gauge correction technique show significant improvement in reducing these errors when compared to the naive implementation. Most notably, the gauge correction technique ensures the satisfaction of the gauge condition at the level of machine precision.

\subsection{Drifting Cloud of Electrons}

Given the task of minimizing the gauge error, a test problem that tends to have a high gauge error was needed. When charged particles move through neutral space, the system is prone to gauge errors unless properly managed. As such, we take a periodic domain and induce a potential well by placing a grouping of ions, which are normally distributed in space and stationary, at its center. The electrons are given an identical distribution in space, but are given a drift velocity equal to $c/100$. With such a drift, the electrons are able to escape the potential well created by the ions and move throughout the domain (see Figure \ref{fig:moving-cloud-scatter-plots}). The domain, itself, is a periodic square $\left[-1/2 ,1/2\right]^{2}$ given in units of $\lambda_{D}$, and we run the simulation to a final time of $T = 1/2$ in units of $\omega_{pe}^{-1}$. The specific values for the simulation parameters used in our experiment are specified in Table \ref{tab:cloud plasma parameters}. The normalized permittivity and normalized permeability are given by $\sigma_{1} = 1$ and $\sigma_{2} = 1.686555 \times 10^{-6}$, respectively, and the normalized speed of light is $\kappa = 7.700159 \times 10^{2}$.

\begin{table}[!htb]
    \centering
    \def\arraystretch{1.2}
    \begin{tabular}{ | c || c | }
        \hline
        \textbf{Parameter}  & \textbf{Value} \\
        \hline
        Average number density ($\bar{n}$) [m$^{-3}$] & $1.0\times 10^{13}$ \\
        Average electron temperature ($\bar{T}$) [K] & $1.0\times 10^{4}$ \\
        Debye length ($\lambda_{D}$) [m] & $2.182496 \times 10^{-3}$ \\
        Electron angular plasma period ($\omega_{pe}^{-1}$) [s/rad] & $5.605733\times 10^{-9}$ \\
        Electron drift velocity ($v_{d}^{(1)} = v_{d}^{(2)}$) [m/s] & $c/ 100$ \\
        Electron thermal velocity ($v_{th} = \lambda_D  \omega_{pe}$) [m/s] & $3.893328 \times 10^{5}$ \\
        \hline
    \end{tabular}
    \caption{Plasma parameters used in the simulation of the drifting cloud of electrons. }
    \label{tab:cloud plasma parameters}
\end{table}

In Figure \ref{fig:Cloud Gauge Error}, we plot the gauge error for the drifting cloud problem and compare the proposed methods against a naive implementation that uses bilinear maps for both the charge density and current density. In each of the methods, we examined the gauge error on spatial meshes with different resolution, with the coarsest mesh being $16 \times 16$ and the finest being $64 \times 64$. Aside from the naive approach, we note that the gauge error for a given method increases proportionally with the resolution of the spatial discretization. One explanation for this is that coarser spatial discretizations are less likely to resolve the gauge error, which is generally small, compared to more refined meshes. Using bilinear interpolation for both the charge density and current density results in a gauge error that starts at $\mathcal{O}\left(10^{-5}\right)$ and quickly shrinks to $\mathcal{O}\left(10^{-7}\right)$. On the other hand, we find that using a continuity equation to define the charge density, in combination with a bilinear mapping for the current density, results in a much lower error which is $\mathcal{O}\left(10^{-13}\right)$. This is true for methods which compute the numerical derivatives using either the FFT or the sixth-order finite difference method. We remark that the method used for differentiation must be identical to the one used for the potentials, following the discussion in section \ref{subsec:Remark on derivatives}.  In practice, solving the system with this technique results in an accumulation of error that is close to machine precision, yielding $\mathcal{O}(10^{-13})$ accuracy. The gauge correction technique, on the other hand, does not accumulate this error, so we observe a much smaller error of $\mathcal{O}(10^{-18})$.

\begin{figure}[!htb]
    \centering
    \subfloat[$t = 0$]{
    \includegraphics[clip, trim={0cm, 0cm, 0cm, 0cm}, scale=0.23]{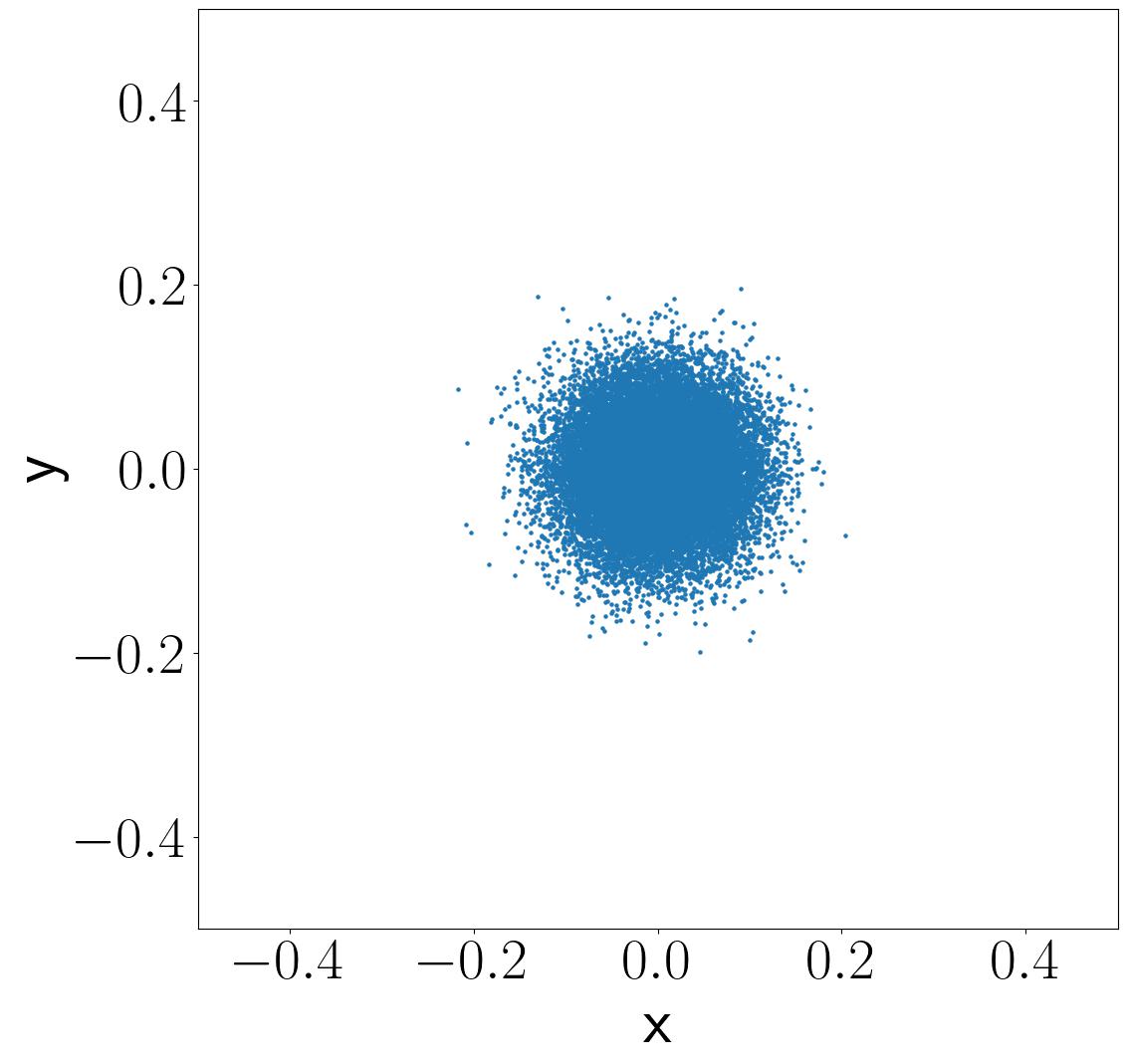}}
    \subfloat[$t = 0.05$]{
    \includegraphics[clip, trim={0cm, 0cm, 0cm, 0cm}, scale=0.23]{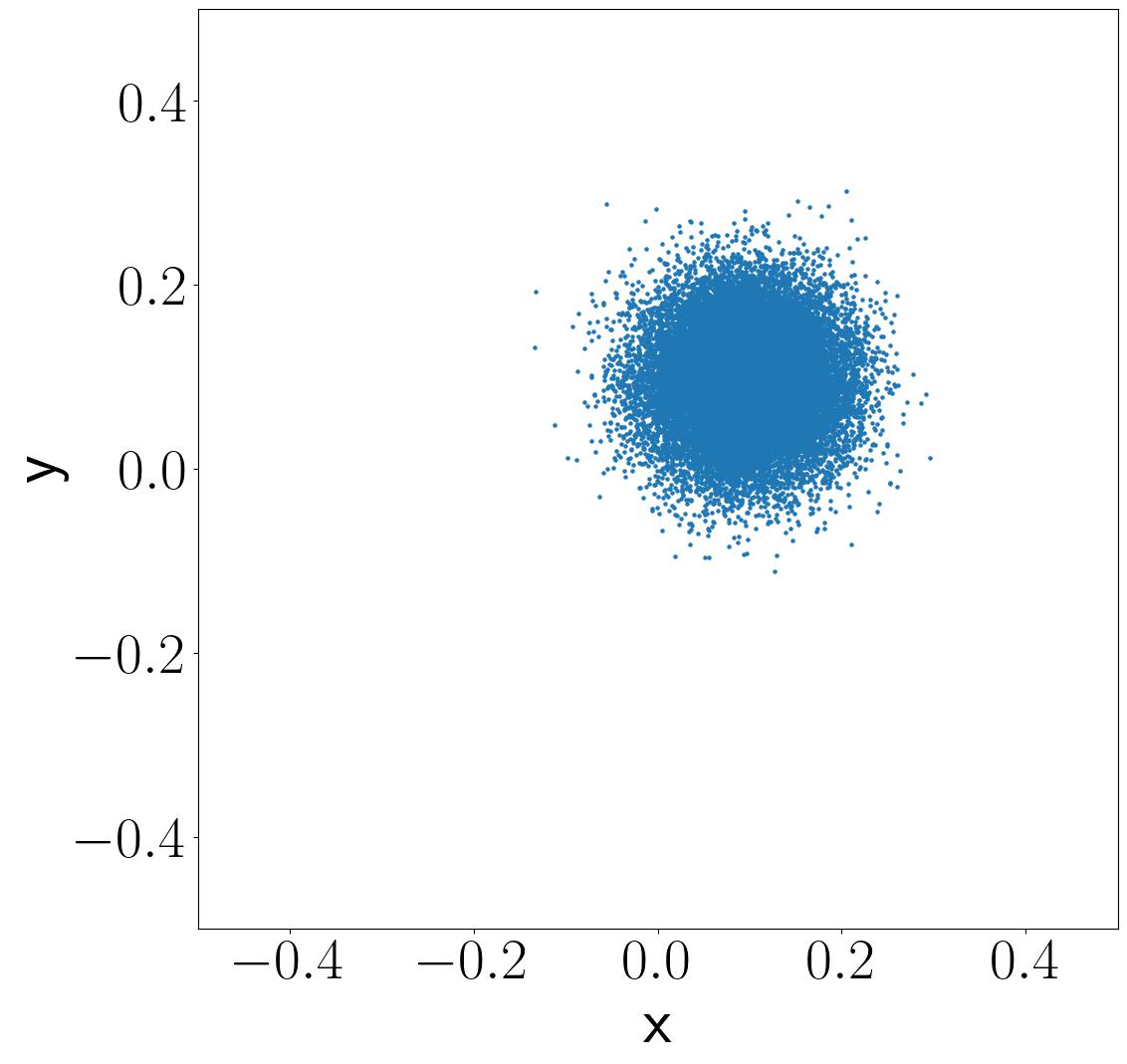}}
    
    \subfloat[$t = 0.1$]{
    \includegraphics[clip, trim={0cm, 0cm, 0cm, 0cm}, scale=0.23]{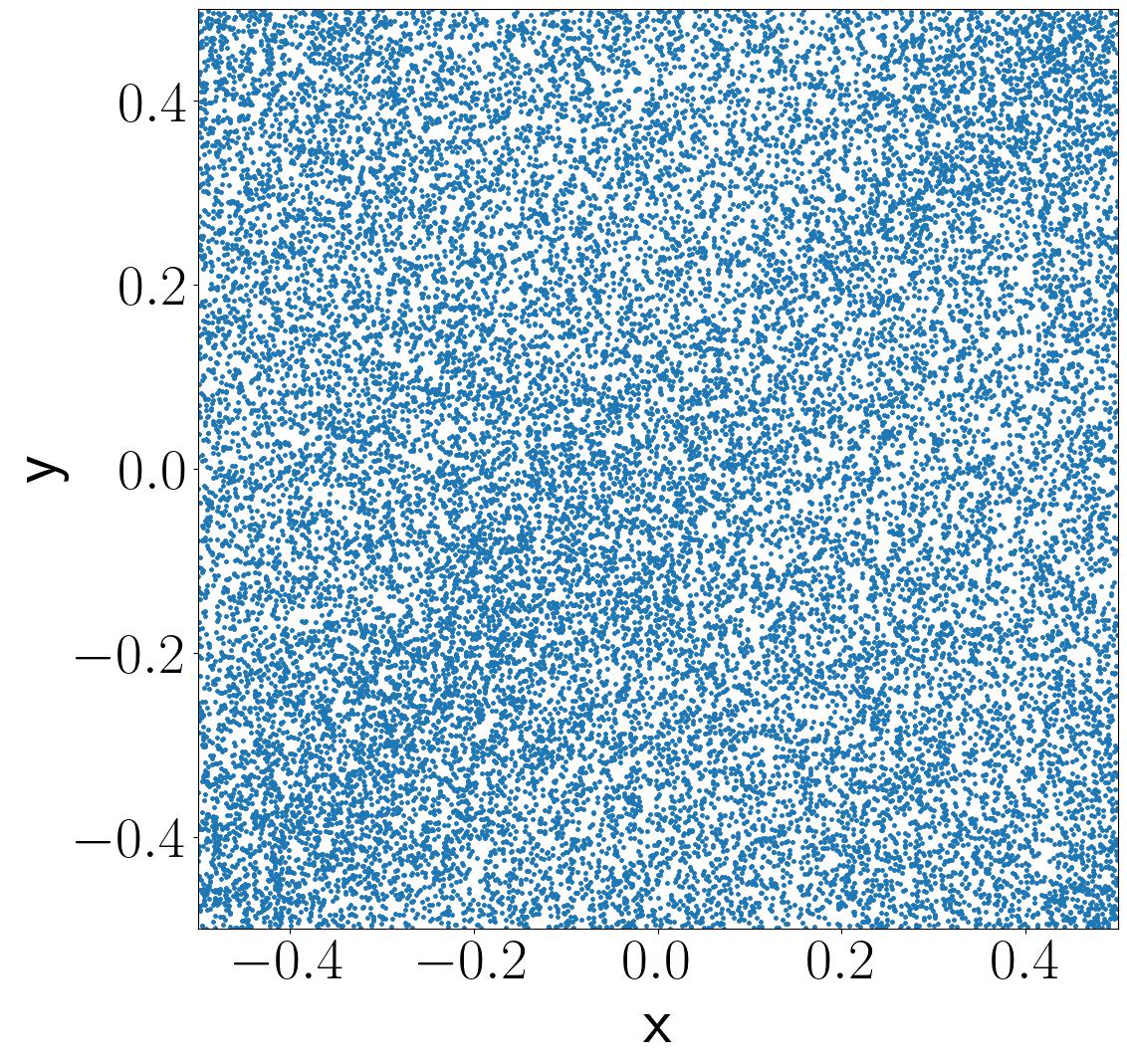}}
    \subfloat[$t = 0.5$]{
    \includegraphics[clip, trim={0cm, 0cm, 0cm, 0cm}, scale=0.23]{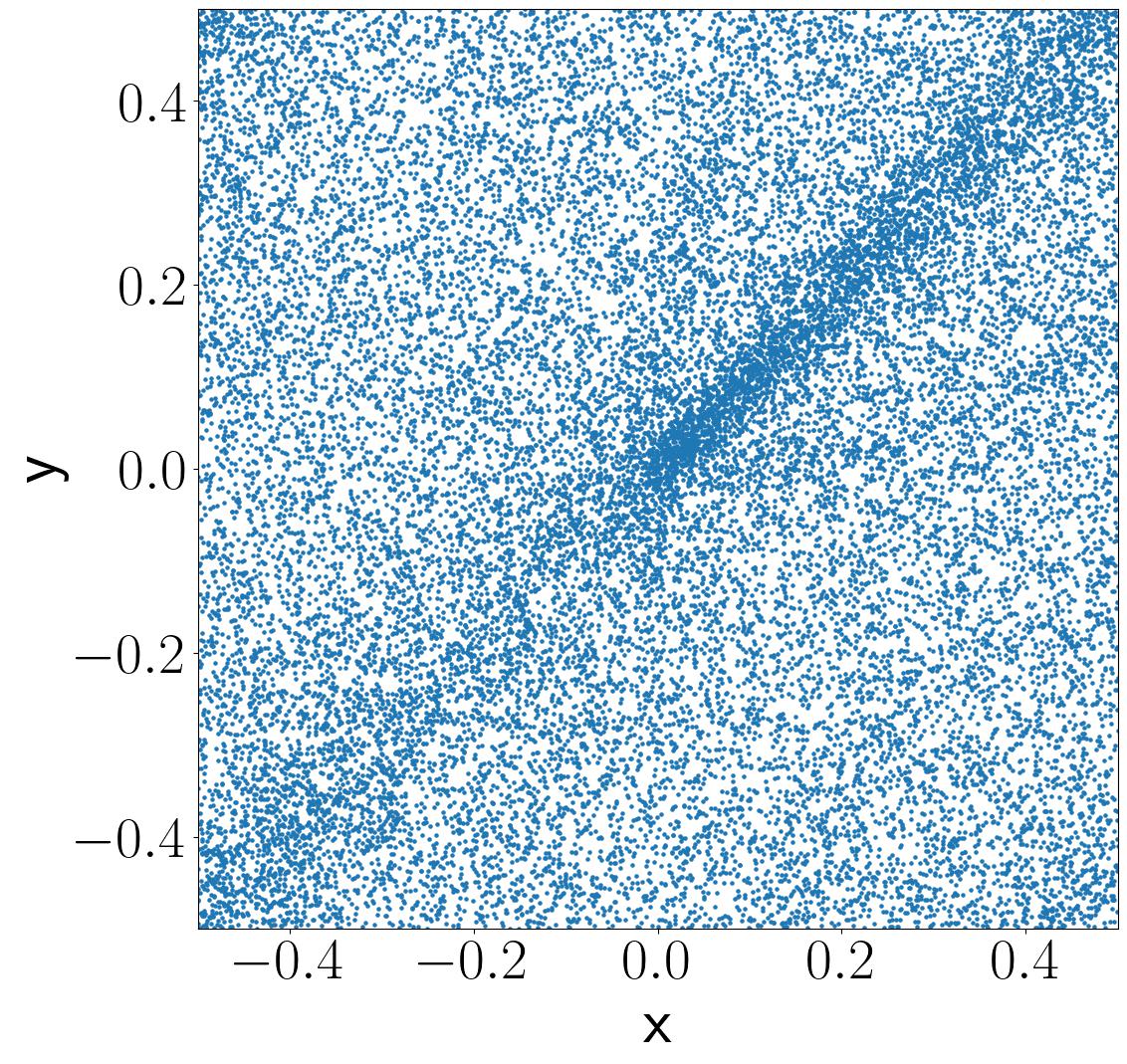}}
    \caption{A Maxwellian distribution of electrons and ions are placed on a periodic domain. The electrons are given a thermal velocity in addition to a drift velocity $\textbf{v}_{d} = (v_{d}^{(1)}, v_{d}^{(2)})^T$, where $v^{(1)} = v^{(2)} = c/100$, that is the speed is $\lVert \textbf{v}_{d} \rVert = \sqrt{2}c/100$. We see the particles escape the well only to fall back in later. The plot (b) at $t=0.05$ is their second such traversal. By $t=0.1$ the electrons have become quite dispersed throughout the domain, though the remnants of the original drift velocity can be found in a slight beam along the diagonal ($t=0.5)$.}
    \label{fig:moving-cloud-scatter-plots}
\end{figure}

\begin{figure}[!htb]
    \centering
    \subfloat[Bilinear Map]{
    \includegraphics[clip, trim={0cm, 0cm, 0cm, 0cm}, scale=0.23]{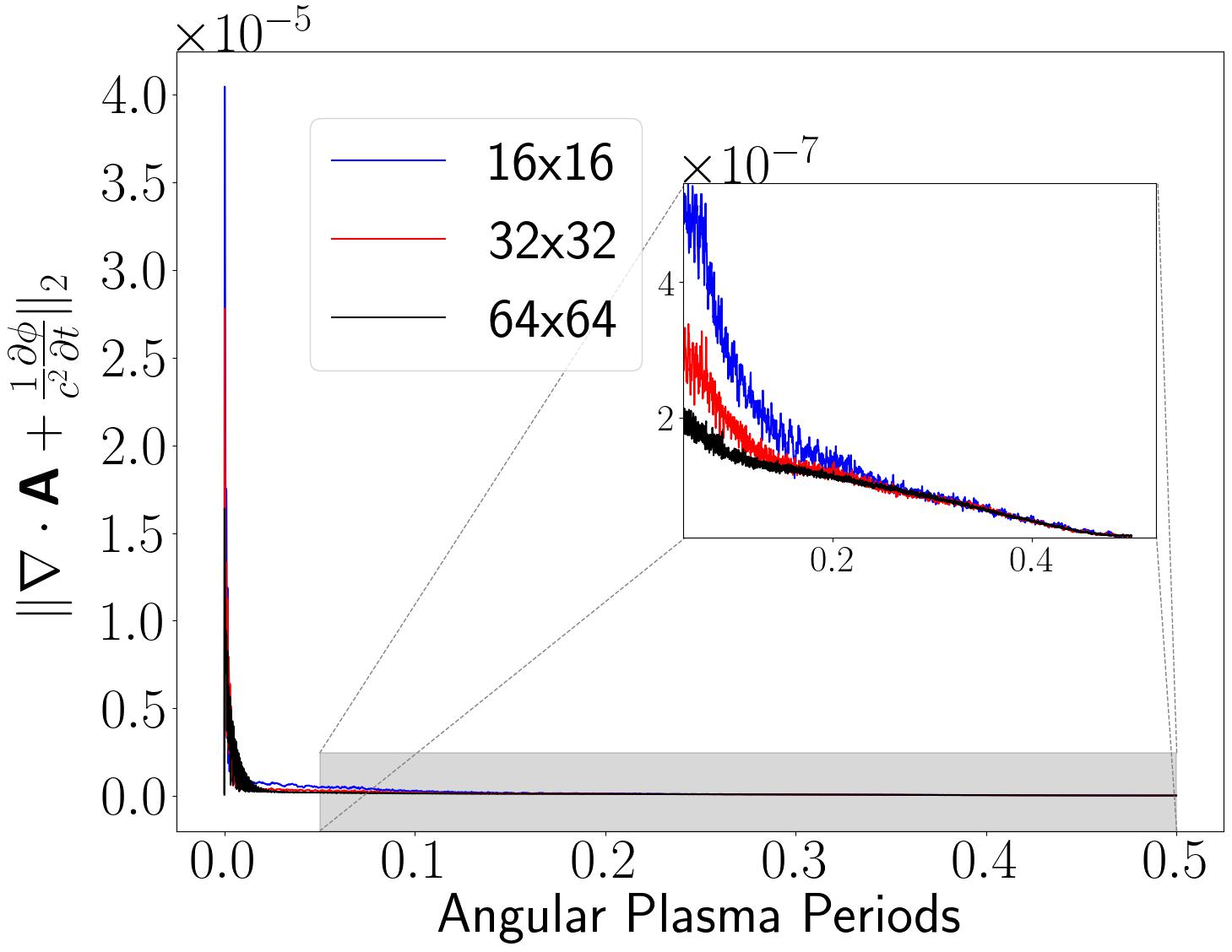}}    
    \subfloat[FFT Satisfying Continuity Equation]{
    \includegraphics[clip, trim={0cm, 0cm, 0cm, 0cm}, scale=0.23]{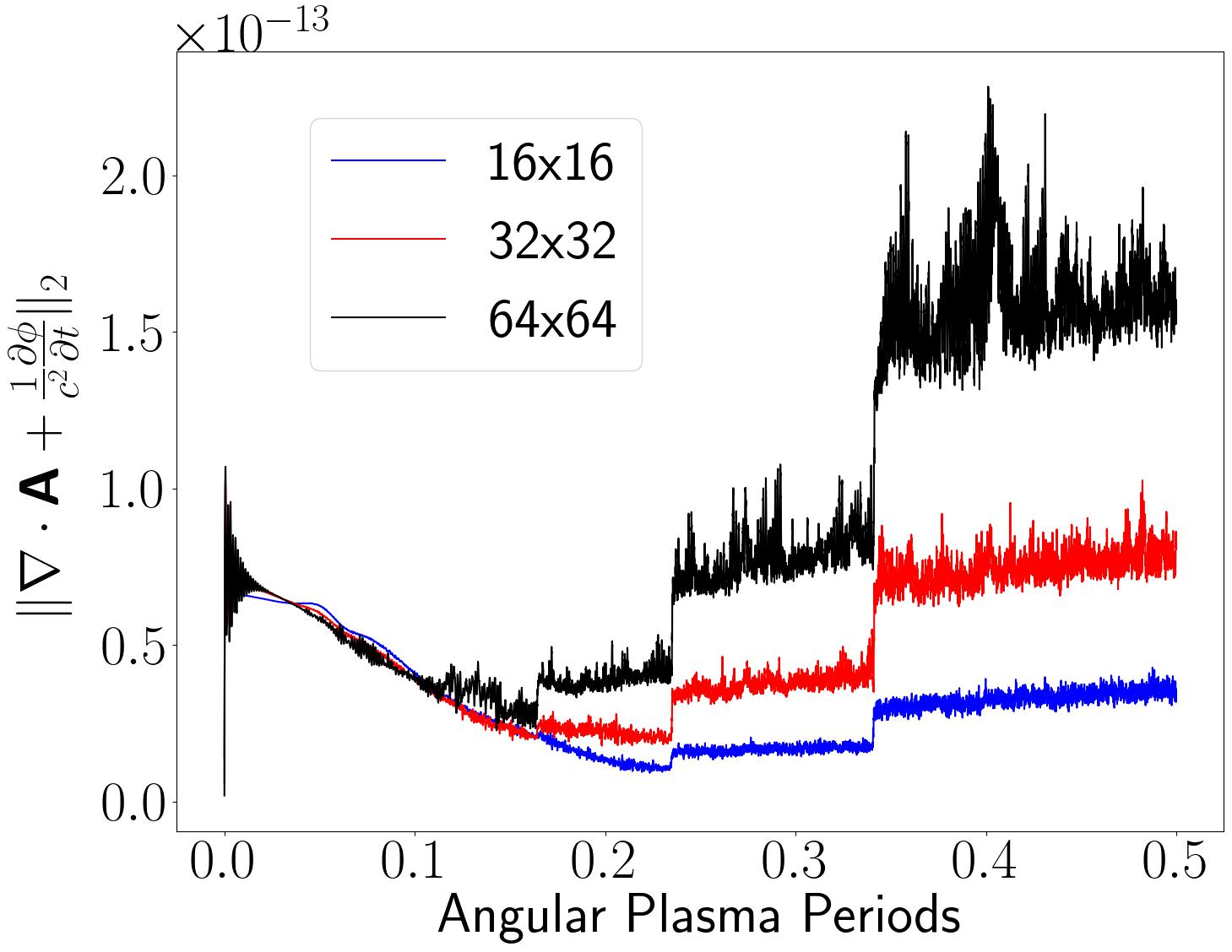}}
    
    \subfloat[FD6 Satisfying Continuity Equation]{
    \includegraphics[clip, trim={0cm, 0cm, 0cm, 0cm}, scale=0.23]{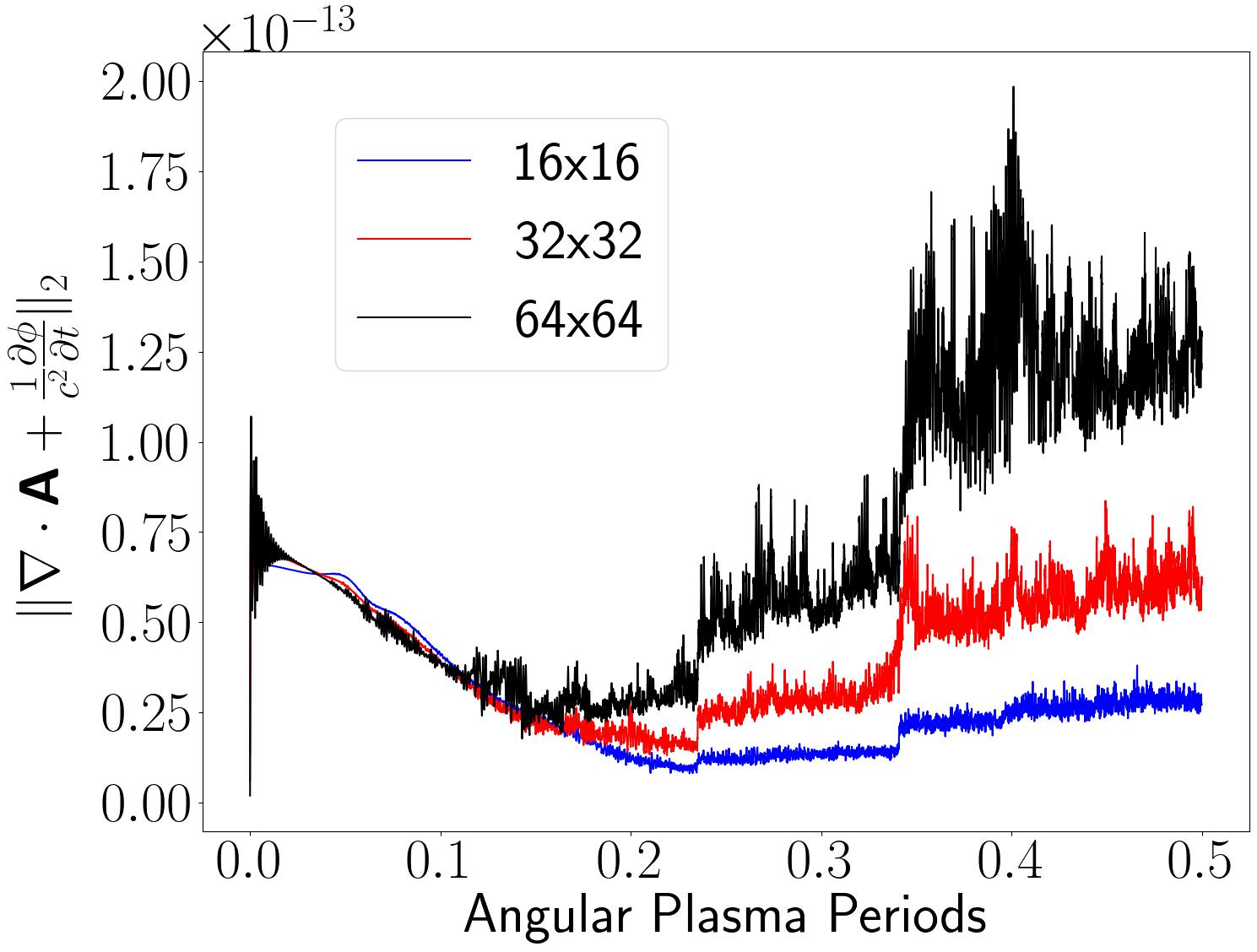}}    
    \subfloat[Gauge Correcting]{
    \includegraphics[clip, trim={0cm, 0cm, 0cm, 0cm}, scale=0.23]{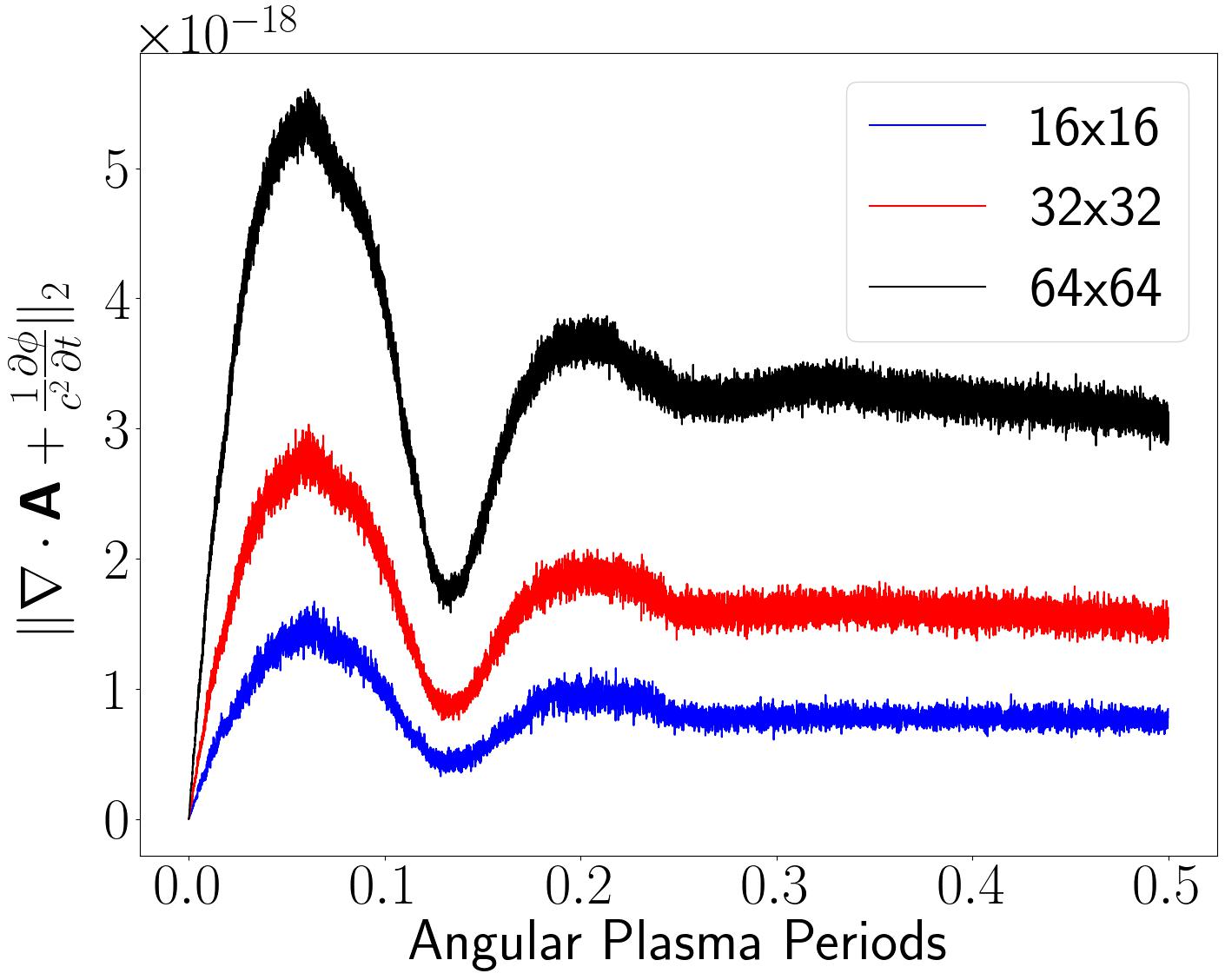}}
    \caption{The gauge error of a cloud of electrons drifting into and out of a potential well. Naively interpolated such that the charge and current densities are not consistent, we see a significant gauge error. However, if the charge density is computed from the current density using the continuity equation and a high order method (eg FFT) to compute the divergence of $\textbf{J}$, we see a much better gauge error over time.}
    \label{fig:Cloud Gauge Error}
\end{figure}

%
%
\subsection{Summary}
\label{subsec:5 Summary}

In this section, we presented some numerical results to compare with the theoretical properties established in section \ref{sec:2 formulation}. We compared several of the proposed methods for enforcing the Lorenz gauge condition against a naive approach in which bilinear maps were used for both the charge and the current density. In our experiments, we considered two test problems defined on periodic domains, namely the relativistic Weibel instability and a drifting cloud of electrons. In the case of the Weibel instability, we found that the growth rate in the magnetic field produced by each of the methods demonstrated good agreement with the predicted growth rate obtained from linear response theory. The second problem we consider is designed to induce a gauge error. In each of the examples, we find that the proposed methods offer notable improvements in reducing the error in the gauge condition when compared to the naive approach.

%
%
\section{Conclusions}
\label{sec:6 Conclusion}

This paper extended the PIC method developed in our previous paper \cite{christlieb2024pic} by investigating new methods to enforce the Lorenz gauge condition. A key result of this paper is Theorem \ref{thm:consistency in time}, which established a connection between the semi-discrete Lorenz gauge condition \eqref{eq:semi-discrete Lorenz} and a corresponding semi-discrete continuity equation \eqref{eq:semi-discrete continuity}. It was also proved that the satisfaction of the gauge condition implies the satisfaction of Gauss' law \eqref{eq:Gauss-E} (Theorem \ref{thm:lorenz implies gauss}). Guided by these semi-discrete theorems, we proposed three methods to enforce the Lorenz gauge condition: two charge conserving maps and one gauge correction technique. With these techniques in hand, we investigated their capabilities in periodic domains. The first problem we considered was the relativistic Weibel instability, in which the methods could be compared against analytical theory. The second test problem simulated a moving cloud of electrons and was specifically designed to induce a gauge error. In both problems, the charge conserving methods yielded a substantial improvement in reducing the error in the gauge condition when compared against the naive implementation.

The methods presented in this work can be further developed in a number of ways. First, we plan to extend these methods to the setting of a bounded domain. In more realistic applications, particles can be injected into a domain or emitted from a hot surface. Particles can also be absorbed or ``stick" to the surface of a conducting material, which will require a careful treatment of the boundary conditions. Additionally, the algorithms presented in this work rely on Cartesian grids, but these techniques should be further extended to handle geometry in several ways, including curvilinear coordinate transformations or even finite-element discretizations. As discussed earlier, we are especially interested in combining the map introduced in section \ref{subsec:3-charge-conserving-boundary-integral-solution} with a boundary integral formulation such as \cite{cheng2017asymptotic}. A more obvious extension of these methods is the extension to higher-order accuracy in time through either a time staggered formulation or an implicit PIC method. The latter is a non-trivial extension, but offers a number of benefits including symplecticity and the ability to take larger time steps in the particle advance. This last feature will be relevant for problems where the ion dynamics become important. It would also be interesting to introduce collision operators into the formulation. A recent paper \cite{bailo2024collisional} showed that it is possible to include grazing collisions from the Landau collision operator in PIC methods in a way that satisfies important collision invariants and guarantees the dissipation of a regularized entropy functional.
%
%
\section{Acknowledgements}

Some of the simulations presented in this work were supported by computational resources made available through the Institute for Cyber-Enabled Research at Michigan State University. The authors would like to thank both the Air Force Office of Scientific Research, the National Science Foundation, and the Department of Energy for their support though grants FA9550-19-1-0281,  FA9550-17-1-0394,  DMS-1912183, and DE-SC0023164.

\printbibliography

@article{VillasenorChargeConservation92,
    title={Rigorous charge conservation for local electromagnetic field solvers},
    author={Villasenor, John and Buneman, Oscar},
    journal={Computer Physics Communications},
    volume={69},
    issue={2-3},
    pages={306-316},
    year={1992}
}

@article{Chen-ImplicitPIC-Darwin2014,
author = {G. Chen and L. Chac\'on},
title = {An energy- and charge-conserving, nonlinearly implicit, electromagnetic {1D-3V} {V}lasov–{D}arwin particle-in-cell algorithm},
journal = {Computer Physics Communications},
volume = {185},
number = {10},
pages = {2391-2402},
year = {2014}
}

@article{Weibel1959,
    author = {Weibel, E.S.},
    title = {Spontaneously Growing Transverse Waves in a Plasma Due to an Anisotropic Velocity Distribution},
    journal = {Physical Review Letters},
    year = {1959},
    volume = {2},
    issue = {3},
    pages = {83-84}
}

@article{MorseWeibel1971,
    author = {Morse, R.L. and C.W. Nielson},
    title = {Numerical Simulation of the {W}eibel Instability in One and Two Dimensions},
    journal = {Physics of Fluids},
    year = {1971},
    volume = {14},
    pages = {830-840}
}

@article{yoon1987exact,
    title={Exact analytical model of the classical {W}eibel instability in a relativistic anisotropic plasma},
    author={Yoon, Peter H. and Davidson, Ronald C.},
    journal={Physical Review A},
    volume={35},
    number={6},
    pages={2718},
    year={1987},
    publisher={APS}
}

@article{FonsecaWeibel03,
    author = {Fonseca, R.A. and Silva, L.O. and Tonge, J.W. and Mori, W.B. and Dawson, J.M.},
    title = {Three-dimensional {W}eibel instability in astrophysical scenarios},
    journal = {Physics of Plasmas},
    year = {2003},
    volume = {10},
    number = {5},
    pages = {1979–1984}
}

@article{LeeFilamentation1973,
    author = {Roswell Lee and Martin Lampe},
    title = {Electromagnetic Instabilities, Filamentation, and Focusing of Relativistic Electron Beams},
    journal = {Physical Review Letters},
    volume = {31},
    number = {3},
    year = {1973},
    doi = {10.1103/PhysRevLett.31.1390}
}

@article{WeibelAstro2021,
  title = {Magnetic Field Amplification by the {W}eibel Instability at Planetary and Astrophysical Shocks with High Mach Number},
  author = {Bohdan, Artem and Pohl, Martin and Niemiec, Jacek and Morris, Paul J. and Matsumoto, Yosuke and Amano, Takanobu and Hoshino, Masahiro and Sulaiman, Ali},
  journal = {Phys. Rev. Lett.},
  volume = {126},
  issue = {9},
  pages = {095101},
  numpages = {6},
  year = {2021},
  month = {3},
  publisher = {American Physical Society},
  doi = {10.1103/PhysRevLett.126.095101}
}

@article{FluidKineticICFsims2005,
    author = {S. Atzeni and A. Schiavi and F. Califano and F. Cattani and F. Cornolti and D. {Del Sarto} and T.V. Liseykina and A. Macchi and F. Pegoraro},
    title = {Fluid and kinetic simulation of inertial confinement fusion plasmas},
    journal = {Computer Physics Communications},
    volume = {169},
    number = {1},
    pages = {153--159},
    year =  {2005},
    doi = {10.1016/j.cpc.2005.03.036}
}

@article{Gibbon2010,
	title={Mesh-Free Magnetoinductive Plasma Model},
	author={Ma\"{u}sek, M. and Gibbon, P.},
	journal={IEEE Transactions on Plasma Science},
	volume={38},
	issue={9},
	pages={2377-2382},
	year={2010},
	publisher={IEEE}
}

@article{Gibbon2017Hamiltonian,
	title={Mesh-free {H}amiltonian implementation of two dimensional {D}arwin model},
	author={Siddi, L. and Lapenta, G. and Gibbon, P.},
	journal={Physics of Plasmas},
	volume={24},
	issue={8},
	pages={1-11},
	year={2017},
	publisher={American Institute of Physics}
}

@article{causley2014higher,
  title={Higher order {A}-stable schemes for the wave equation using a successive convolution approach},
  author={Causley, Matthew F and Christlieb, Andrew J},
  journal={SIAM Journal on Numerical Analysis},
  volume={52},
  number={1},
  pages={220--235},
  year={2014},
  publisher={SIAM}
}

@article{causley2017wave-propagation,
  title={Method of lines transpose: An efficient unconditionally stable solver for wave propagation},
  author={Causley, Matthew and Christlieb, Andrew and Wolf, Eric},
  journal={Journal of Scientific Computing},
  volume={70},
  number={2},
  pages={896--921},
  year={2017},
  publisher={Springer}
}

@article{cheng2017asymptotic,
	title={An asymptotic preserving {M}axwell solver resulting in the {D}arwin limit of electrodynamics},
	author = {Cheng, Yingda and Christlieb, Andrew J and Guo, Wei and Ong, Benjamin},
	journal = {Journal of Scientific Computing},
	volume = {71},
	number = {3},
	pages = {959--993},
	year = {2017},
	publisher = {Springer}
}

@article{MOLT-EB-2020,
	title={A fast local embedded boundary method suitable for high power electromagnetic sources},
	author={Thavappiragasm, M. and Christlieb, A.J. and Luginsland, J. and Guthrey, P.T.},
	journal={AIP Advances},
	volume={10},
	issue={11},
	pages={115318},
	year={2020},
	publisher={AIP Publishing LLC}
}

@book{FollandBook1995,
  title={Introduction to Partial Differential Equations: Second edition},
  author={Folland, Gerald B.},
  isbn={978-0691213033},
  year={1995},
  publisher={Princeton University Press;}
}

@misc{christlieb2024pic,
      title={A Particle-in-cell Method for Plasmas with a Generalized Momentum Formulation, Part {I}: Model Formulation}, 
      author={Andrew J. Christlieb and William A. Sands and Stephen White},
      year={2024},
      eprint={2208.11291},
      archivePrefix={arXiv},
      primaryClass={physics.plasm-ph}
}

@misc{bailo2024collisional,
      title={The Collisional Particle-In-Cell Method for the {V}lasov-{M}axwell-{L}andau Equations}, 
      author={Rafael Bailo and José A. Carrillo and Jingwei Hu},
      year={2024},
      eprint={2401.01689},
      archivePrefix={arXiv},
      primaryClass={physics.plasm-ph}
}

\end{document}